\newif\ifIEEE
\newif\ifPAGELIMIT
\newif\ifANONYMOUS
\newif\ifARXIV
\ARXIVtrue
\IEEEtrue
\ANONYMOUStrue
\ANONYMOUSfalse
\PAGELIMITfalse
\ifPAGELIMIT
    \IEEEtrue
\else
    \ANONYMOUSfalse
\fi
\ifARXIV
    \IEEEtrue
    \PAGELIMITfalse
\fi
\ifIEEE
    \ifPAGELIMIT
        \documentclass[conference,final,letterpaper]{IEEEtran}
    \else
        \documentclass[journal,final,letterpaper]{IEEEtran}
    \fi
    \newcommand{\bibauthor}[1]{#1}
    \newcommand{\bibpaper}[1]{``#1''}
    \newcommand{\Footnotetext}[2]
    {
        \begin{figure}[!b]
        \footnotesize\vspace{-3ex}\hrulefill\hfill
        \makebox[0em]{}\hfill\makebox[0em]{}%
                                          \par${}^{#1}$ #2\vspace{-.6ex}
        \end{figure}
        \addtocounter{figure}{0}
     }
\else
    \documentclass[12pt]{article}
    \newcommand{\bibauthor}[1]{\textsc{#1}}
    \newcommand{\bibpaper}[1]{\textsl{#1}}
    \newenvironment{IEEEkeywords}{\begin{small}%
                                  \textbf{Index Terms} ---}{\end{small}}
\fi
\usepackage{amsfonts,bm,bbm}
\usepackage{amsthm,amsmath,amssymb}
\usepackage{pict2e}
\usepackage{color}
\newcommand{\bibbook}[1]{\textit{#1}}

\newcommand{\bibperiodical}[1]{\textit{#1}}

\ifPAGELIMIT
\fi

\newtheorem{theorem}{\indent Theorem}
\newtheorem{proposition}[theorem]{\indent Proposition}
\newtheorem{lemma}[theorem]{\indent Lemma}

\theoremstyle{remark}
\newtheorem{remark}{\indent Remark}
\theoremstyle{definition}
\newtheorem{example}{\indent Example}
\renewcommand{\mathbf}[1]{{\bm{#1}}}     
\newcommand{\Emph}[1]{{\underline{#1}}}
\newcommand{\Integers}{{\mathbb{Z}}}

\newcommand{\Array}{{\mathsf{A}}}
\newcommand{\Antichain}{{\mathcal{A}}}
\newcommand{\Boolean}{{\mathbb{B}}}
\newcommand{\weight}{{\mathsf{w}}}
\newcommand{\bldtheta}{{\mathbf{\vartheta}}}
\newcommand{\bldxi}{{\mathbf{\xi}}}
\newcommand{\allzero}{{\mathit{0}}}
\newcommand{\allone}{{\mathit{1}}}

\newcommand{\bldf}{{\mathbf{f}}}
\newcommand{\bldg}{{\mathbf{g}}}
\newcommand{\bldu}{{\mathbf{u}}}
\newcommand{\bldv}{{\mathbf{v}}}
\newcommand{\bldx}{{\mathbf{x}}}
\newcommand{\TCAM}{{\mathrm{T}}}
\newcommand{\E}{{\mathrm{E}}}
\newcommand{\N}{{\mathrm{N}}}
\newcommand{\GE}{\Gamma}
\newcommand{\LE}{\mathrm{L}}
\newcommand{\Eset}{{\mathcal{E}}}
\newcommand{\Nset}{{\mathcal{N}}}
\newcommand{\GEset}{{\mathcal{G}}}
\newcommand{\LEset}{{\mathcal{L}}}
\newcommand{\Fset}{{\mathcal{F}}}
\newcommand{\Collection}{{\mathcal{H}}}
\newcommand{\Sphere}{{\mathcal{S}}}
\newcommand{\U}{{\mathcal{U}}}
\newcommand{\reject}{{\bullet}}
\newcommand{\Sdb}{{${*}{\circ}$}}
\newcommand{\Sbd}{{${\circ}{*}$}}
\newcommand{\Sdd}{{${*}{*}$}}
\newcommand{\Srd}{{$\reject{*}$}}
\newcommand{\Sdr}{{${*}\reject$}}
\newcommand{\Srr}{{$\reject\reject$}}
\newcommand{\Sbr}{{${\circ}\reject$}}
\newcommand{\Sbb}{{${\circ}{\circ}$}}
\newcommand{\Srb}{{$\reject{\circ}$}}
\newcommand{\ldbrack}{\makebox[.9ex]{$[\![$}}
\newcommand{\rdbrack}{\makebox[.9ex]{$]\!]$}}
\ifPAGELIMIT
    \newcommand{\ifandonlyif}{iff}
\else
\newcommand{\ifandonlyif}{if and only if}
\fi

\newcommand{\Title}{On the Implementation of Boolean Functions
            on Content-Addressable Memories}
\newcommand{\Name}{Ron M. Roth}
\newcommand{\Address}{Computer Science Department,
                      Technion,
                      Haifa 3200003, Israel}
\newcommand{\Membership}{, \IEEEmembership{Fellow, IEEE}}
\newcommand{\Email}{ronny@cs.technion.ac.il}
\ifANONYMOUS
    \newcommand{\NameX}{Anonymous}
    \newcommand{\AddressX}{(Address left blank)}
    \newcommand{\EmailX}{(Email left blank)}
\else
\newcommand{\NameX}{\Name}
\newcommand{\AddressX}{\Address}
\newcommand{\EmailX}{\Email}
\fi
\ifPAGELIMIT
    \newcommand{\Grant}{Work supported in part by
                        Grant~1713/20 from
                        the Israel Science Foundation}
    \newcommand{\Addressalt}{Work done in part at
                Hewlett Packard Labs, Milpitas, CA}
\else
\newcommand{\Grant}{This work was supported in part by
                    Grant~1713/20 from
                    the Israel Science Foundation}
\newcommand{\Conference}{
            Part of this work was accepted for presentation
            at the IEEE Int'l Symposium on Information
            Theory (ISIT), June 2023}
\newcommand{\Addressalt}{This work was done in part while
            visiting Hewlett Packard Laboratories,
            Milpitas, CA}
\fi
\newcommand{\Thnx}{\Name\ is with the \Address.
                    \Addressalt.
                    \Grant.
                    \Conference.
                    Email: \Email.
                    \par}

\begin{document}
\ifIEEE
    \title{\Title}
       \ifPAGELIMIT
           \author{\IEEEauthorblockN{\NameX\vspace{-1ex}}\\
                   \IEEEauthorblockA{\AddressX.
                                     Email: \EmailX\vspace{-2ex}}
           }
       \else
           \ifARXIV
               \author{\Name\thanks{\Thnx}}
           \fi
       \fi
\else
    \title{\textbf{\Title}}
    \author{\textsc{\Name}\thanks{\Thnx}}
    \date{}
\fi
\maketitle

\begin{abstract}
Let $[q\rangle$ denote the integer set $\{0,1,\ldots,q-1\}$
and let $\Boolean=\{0,1\}$.
The problem of implementing functions
$[q\rangle\rightarrow\Boolean$
on content-addressable memories (CAMs) is considered.
CAMs can be classified by the input alphabet
and the state alphabet of their cells;
for example, in binary CAMs,
those alphabets are both $\Boolean$, while
in a ternary CAM (TCAM), both alphabets are endowed with
a ``don't care'' symbol.

This work is motivated by recent proposals for using CAMs
for fast inference on decision trees.
In such learning models, the tree nodes carry out
integer comparisons, such as testing equality ($x=t$~?)
or inequality ($x\le t$~?),
where $x\in[q\rangle$ is an input to the node
and $t\in[q\rangle$ is a node parameter.
A CAM implementation of such comparisons
includes mapping (i.e., encoding) $t$ into internal states of
some number $n$ of cells
and mapping $x$ into inputs to these cells,
with the goal of minimizing $n$.

Such mappings are presented for various comparison families,
as well as for the set of all functions
$[q\rangle\rightarrow\Boolean$, under several scenarios of
input and state alphabets of the CAM cells.
All those mappings are shown to be optimal in that
they attain the smallest possible $n$ for any given $q$.
\end{abstract}

\ifPAGELIMIT
\else
\fi

\ifPAGELIMIT
    \ifANONYMOUS
    \else
    \Footnotetext{\quad}{\Addressalt. \Grant.}
    \fi
\fi

\section{Introduction}
\label{sec:introduction}

For $a, b \in \Integers$, denote by $[a:b]$ the integer subset
$\{ z \in \Integers \,:\, a \le z \le b \}$
and by $[a:b \rangle$ the set $[a:b-1]$;
we will use the shorthand notation $[b \rangle$ for $[0:b \rangle$.

Let~$\Boolean$ denote the Boolean set $\{ 0, 1 \}$
and let $\Boolean_*$ and $\Boolean_\reject$ denote
the alphabets $\Boolean \cup \{ {*} \}$
and $\Boolean \cup \{ {*}, \reject \}$, respectively.
We refer to $*$ as the ``don't-care'' symbol
and to $\reject$ as the ``reject'' symbol.
Define the function
$\TCAM : \Boolean_\reject^2 \rightarrow \Boolean$
by the truth table shown in Table~\ref{tab:T}.
\begin{table}[hbt]
\caption{Truth table of
the function $(u,\vartheta) \mapsto \TCAM(u,\vartheta)$.%
\ifPAGELIMIT
    \ifANONYMOUS
    \else
    \vspace{-2ex}
    \fi
\fi}
\label{tab:T}
\[
\ifPAGELIMIT
    \renewcommand{\arraystretch}{0.9}
    \arraycolsep0.9ex
\else
\renewcommand{\arraystretch}{1.1}
\arraycolsep1ex
\fi
\begin{array}{c|cccc}
\raisebox{-.2ex}{$u$}\backslash\raisebox{0.2ex}{$\vartheta$}
        & \reject & 0 & 1 & * \\
\hline
\raisebox{0ex}[2.2ex]{}
\reject & 0       & 0 & 0 & 1 \\
0       & 0       & 1 & 0 & 1 \\
1       & 0       & 0 & 1 & 1 \\
*       & 1       & 1 & 1 & 1
\end{array}
\]
\ifPAGELIMIT
    \ifANONYMOUS
    \else
    \vspace{-2ex}
    \fi
\fi
\end{table}
The first argument, $u$, of~$\TCAM$ is called the \emph{input}
and the second argument, $\vartheta$, is the \emph{internal state}.
While the function $(u,\vartheta) \mapsto \TCAM(u,\vartheta)$
is symmetric in its two
arguments (in that $\TCAM(u,\vartheta) = \TCAM(\vartheta,u)$),
there is a distinction between them in the way this function
will be used.

We extend the definition of the function~$\TCAM$ to words (vectors)
over $\Boolean_\reject$ as follows:
for $\bldu = (u_j)_{j \in [n \rangle}$
and $\bldtheta = (\vartheta_j)_{j \in [n \rangle}$
in $\Boolean_\reject^n$,
we define $\TCAM(\bldu,\bldtheta)$ by
\begin{equation}
\label{eq:TCAMconjunction}
\ifPAGELIMIT
    \textstyle
\fi
\TCAM(\bldu,\bldtheta)
= \bigwedge_{j \in [n \rangle} \TCAM(u_j,\vartheta_j) ,
\end{equation}
where $\wedge$ denotes Boolean conjunction (i.e., logical ``and'').

A \emph{content-addressable memory} (CAM) is a device which
consists of an $m \times \ell$ array of \emph{cells}, with each cell
$(i,j) \in [m \rangle \times [\ell \rangle$ implementing
the function $u \mapsto \TCAM(u,\vartheta_{i,j})$,
for some internal state value $\vartheta_{i,j} \in \Boolean_\reject$.
The input to the array is a word
$\bldu = (u_j)_{j \in [\ell \rangle} \in \Boolean_\reject^\ell$,
with $u_j$ serving as the input to all the cells along column~$j$.
The internal states need to be \emph{programmed} into the array
prior to its use and are therefore assumed to change
much less frequently than the input word~$\bldu$.
In addition, the selections of the internal state values and of
the input words are assumed to be independent.

Each row (``match line'') in the array computes the conjunction
of the outputs of the cells along
the row (as in~(\ref{eq:TCAMconjunction})),
and the results obtained by the $m$~rows form
the output word, $\bldv \in \Boolean^m$, of the array.
The nonzero entries in~$\bldv$ are referred to as ``matches,''
and in practice there are settings where one might be interested
only in the number of matches (namely, the Hamming weight of~$\bldv$)
or in the index~$i$ of the first match in~$\bldv$.

CAMs are classified, \emph{inter alia}, by the ranges
of the input values $u_j$
and the internal state values $\vartheta_{i,j}$.
In a binary CAM (in short, BCAM),
these values are constrained to be in~$\Boolean$,
and in a ternary CAM (in short, TCAM),
they are elements of $\Boolean_*$.
BCAMs are used in implementing fast look-up tables,
and TCAMs are used in networking equipment for performing
high-speed packet classification.

There are several common hardware designs for CAM cells,
e.g., \cite{GG},\cite{KTFY},\cite{YSMB}.
Most designs allow in effect the cells (including BCAM cells)
to take both~$*$ and~$\reject$ as input values.
Moreover, some of the TCAM designs, such as
the one shown in~\cite[Fig.~2(a)]{GG} and~\cite[Fig.~1]{KTFY},
allow the internal state to take the value~$\reject$
(hence our election
to define~$\TCAM$ to have the domain $\Boolean_\reject^2$).

For an integer $q \ge 2$,
let $\Fset_q$ denote the set of all functions
$f : [q \rangle \rightarrow \Boolean$
\ifPAGELIMIT
    (when $q = 2^h$, one can regard $\Fset_q$
    as the set of all $h$-variate Boolean functions
    $\Boolean^h \rightarrow \Boolean$).
\else
(to avoid trivialities, we exclude the case
$q = 1$).\footnote{%
When $q = 2^h$, one can regard $\Fset_q$
as the set of all $h$-variate Boolean functions
$\Boolean^h \rightarrow \Boolean$ (by simply replacing
each element in the domain $[q \rangle$ by its
length-$h$ binary representation).}
\fi
For certain function families $\Phi \subseteq \Fset_q$,
we are interested in the implementation of each
$x \mapsto f(x) \in \Phi$ in the form
\begin{equation}
\label{eq:TCAM}
\ifPAGELIMIT
    \textstyle
\fi
f(x) = \TCAM(\bldu(x),\bldtheta(f))
\stackrel{\textrm{(\ref{eq:TCAMconjunction})}}{=}
\bigwedge_{j \in [n \rangle} \TCAM(u_j(x),\vartheta_j(f)) ,
\end{equation}
where
$\bldu : [q \rangle \rightarrow \Boolean_\reject^n$
and
$\bldtheta : \Phi \rightarrow \Boolean_\reject^n$
are prescribed mappings which may (only) depend on~$\Phi$ and~$n$;
in fact, for a given
family~$\Phi$, we will be interested in the smallest possible~$n$
for which any function $f \in \Phi$ can be expressed
in the form~(\ref{eq:TCAM}).
Note that the first argument of $\TCAM(\cdot,\cdot)$ in~(\ref{eq:TCAM})
depends only on~$x$, while the second depends only on~$f$.
Such a separation is dictated by the common use
of CAMs: the functions~$f$ are pre-specified
and the respective words $\bldtheta(f)$ are programmed
into (the internal state values of) $n$~CAM cells along a row,
while~$x$ is an input within a sequence of
inputs to which the functions are to be applied.

Our study of implementations
of Boolean functions on CAMs is motivated primarily
by certain models of machine learning---specifically,
by recent proposals for using CAMs for fast
inference on decision trees~\cite{LGSMFPS},\cite{PGLSSFMS}.
In such learning models, the basic operations carried out
in the nodes of the trees are comparisons between
(real or) integer numbers.
The implementation in~\cite{LGSMFPS},\cite{PGLSSFMS}
can be seen as a generalized notion of a CAM
which consists of an $m \times \ell$ array $\Array$ of cells,
with each cell $(i,s) \in [m \rangle \times [\ell \rangle$
computing a function (comparator)
$x \mapsto f_{i,s}(x) \in \Fset_q$
which tests~$x$ against some fixed threshold.
The input to such an array is
a word $\bldx = (x_s)_{s \in [\ell \rangle} \in [q \rangle^\ell$
and the output of each row~$i$ is
the conjunction of the outputs of the cells along the row:
\begin{equation}
\label{eq:fconjunction}
\ifPAGELIMIT
    \textstyle
\fi
\bigwedge_{s \in [\ell \rangle} f_{i,s}(x_s) .
\end{equation}
When used to implement decision trees, the entries of~$\bldx$
are the features (attributes) and each row~$i$ corresponds
to (a path that leads to) a leaf in the tree.
The expression~(\ref{eq:fconjunction})
evaluates to~$1$ \ifandonlyif\ the feature values in~$\bldx$
lead to the (unique) leaf that corresponds to row~$i$.

Our families~$\Phi$ of interest have been selected
accordingly and, in order to define them,
we introduce the following notation.
Let the bivariate functions
$\E_q$, $\N_q$, $\GE_q$, and $\LE_q$,
all with the domain $[q \rangle^2$
and range~$\Boolean$, be defined
for every $(x,t) \in [q \rangle^2$ by:
\ifPAGELIMIT
\[
    \begin{array}{lclclcl}
    \E_q(x,t)
    & = & \ldbrack x = t \rdbrack ,
    &&
    \GE_q(x,t)
    & = & \ldbrack x \ge t \rdbrack, \\
    \N_q(x,t)
    & = & \ldbrack x \ne t \rdbrack,
    &&
    \LE_q(x,t)
    & = & \ldbrack x \le t \rdbrack ,
    \end{array}
\]
\else
\begin{eqnarray*}
\E_q(x,t)
& = & \ldbrack x = t \rdbrack ,  \\
\N_q(x,t)
& = & \ldbrack x \ne t \rdbrack, \\
\GE_q(x,t)
& = & \ldbrack x \ge t \rdbrack, \\
\LE_q(x,t)
& = & \ldbrack x \le t \rdbrack ,
\end{eqnarray*}
\fi
where $\ldbrack \cdot \rdbrack$ denotes the Iverson bracket
(which evaluates to~$1$ if its argument is true,
and to~$0$ otherwise).
Respectively, define the following subsets of $\Fset_q$
(each of size~$q$):
\ifPAGELIMIT
\[
    \arraycolsep0.6ex
    \begin{array}{lclclcl}
    \Eset_q & = &
    \bigl\{ x \mapsto \E_q(x,t) \bigr\}_{t \in [q \rangle} ,
    &&
    \GEset_q & = &
    \bigl\{ x \mapsto \GE_q(x,t) \bigr\}_{t \in [q \rangle} , \\
    \Nset_q & = &
    \bigl\{ x \mapsto \N_q(x,t) \bigr\}_{t \in [q \rangle} ,
    &&
    \LEset_q & = &
    \bigl\{ x \mapsto \LE_q(x,t) \bigr\}_{t \in [q \rangle} .
\end{array}
\]
\else
\begin{eqnarray*}
\Eset_q & = &
\bigl\{ x \mapsto \E_q(x,t) \bigr\}_{t \in [q \rangle} ,  \\
\Nset_q & = &
\bigl\{ x \mapsto \N_q(x,t) \bigr\}_{t \in [q \rangle} ,  \\
\GEset_q & = &
\bigl\{ x \mapsto \GE_q(x,t) \bigr\}_{t \in [q \rangle} , \\
\LEset_q & = &
\bigl\{ x \mapsto \LE_q(x,t) \bigr\}_{t \in [q \rangle} .
\end{eqnarray*}
\fi
The families~$\Phi$ that we consider are:
\ifPAGELIMIT
    $\Eset_q$, $\Nset_q$, $\GEset_q$
    (or $\LEset_q$), $\GEset_q \cup \LEset_q$,
    as well as the whole set $\Fset_q$.
\else
\[
\Eset_q,\;\;
\Nset_q,\;\;
\GEset_q
\; \;\textrm{(or $\LEset_q$)},
\;\;
\GEset_q \cup \LEset_q,
\]
as well as the whole set $\Fset_q$.\footnote{%
The family $\GEset_q \cup \LEset_q$ is quite prevailing
in decision trees,
where the same feature $x_s$ is typically compared against
both lower and upper thresholds.}
\fi

Once the functions in the selected family~$\Phi$
have an implementation of the form~(\ref{eq:TCAM})
(with mappings $x \mapsto \bldu(x)$ and $f \mapsto \bldtheta(f)$
that attain the smallest possible~$n$),
we can realize an $m \times \ell$ array~$\Array$ of functions
$f_{i,s} \in \Phi$ by an (ordinary) CAM consisting of
an $m \times \ell n$ array~$\tilde{\Array}$
of cells, with each cell
implementing the function $u \mapsto \TCAM(u,\vartheta_{i,j})$
with the internal states $\vartheta_{i,j}$ selected so that
non-overlapping $n$-cell blocks in $\tilde{\Array}$
implement the functions $f_{i,s}$.\footnote{%
The particular architecture of CAMs
that is proposed~\cite{LGSMFPS},\cite{PGLSSFMS},
referred to therein as an analog CAM (a-CAM), consists of cells
which perform comparisons between integers in $[q \rangle$,
for some prescribed~$q$.
The design of the cells makes use of programmable resistors
and the comparisons are carried out in the analog domain.
The reliability of such an a-CAM in practice, however, is yet to be
tested and, in any event, such an a-CAM is expected to operate
only for relatively small values of~$q$.
The alternate approach, which motivates this work, proposes
using more traditional CAM designs (which are well tested and robust)
to realize the same (targeted) functionality of an a-CAM.}
Specifically, we select the internal states in~$\tilde{\Array}$
for each $(i,j) \in [m \rangle \times [\ell n \rangle$ so that
\[
\bigl(
\vartheta_{i,s n} \; \vartheta_{i,s n+1} \; \ldots
\; \vartheta_{i,(s+1)n-1} \bigr) = \bldtheta(f_{i,s}) ,
\]
in which case the output of row~$i$ in~$\Array$
(namely, (\ref{eq:fconjunction})) can be expressed as
\begin{eqnarray*}
\ifPAGELIMIT
    \textstyle
\fi
\bigwedge_{s \in [\ell \rangle}
f_{i,s}(x_s)
& \stackrel{\textrm{(\ref{eq:TCAM})}}{=}
&
\ifPAGELIMIT
    \textstyle
\fi
\bigwedge_{s \in [\ell \rangle}
\TCAM(\bldu(x_s), \bldtheta(f_{i,s})) \\
& = &
\ifPAGELIMIT
    \textstyle
\fi
\bigwedge_{s \in [\ell \rangle}
\bigwedge_{j \in [n \rangle}
\TCAM(u_j(x_s), \vartheta_{i,s n + j}) ,
\end{eqnarray*}
which is the output of row~$i$ in~$\tilde{\Array}$.

Based upon the intended use of CAMs
and the various designs of their cells, we consider
several scenarios, which will be identified
by pairs of symbols from $\{ \circ, {*}, \reject \}$:
the first (respectively, second) symbol in such a pair specifies
the alphabet of the range
of the mapping~$\bldu$ (respectively, $\bldtheta$)
that appears in~(\ref{eq:TCAM}), with $\Boolean_\circ = \Boolean$;
for example, in Scenario~(\Sbd),
the ranges of~$\bldu$ and~$\bldtheta$
are $\Boolean^n$ and $\Boolean_*^n$, respectively.
We will then say that a subset $\Phi \subseteq \Fset_q$
is \emph{$n$-cell implementable} under a given scenario
if there exist mappings~$\bldu$ and~$\bldtheta$
with the appropriate ranges
($\Boolean^n$, $\Boolean_*^n$, or $\Boolean_\reject^n$)
such that every function $f \in \Phi$ can be written in
the form~(\ref{eq:TCAM}).
\ifPAGELIMIT
\else

\subsection{Summary of results}
\label{sec:summary}
\fi

Table~\ref{tab:summary} summarizes our results:
for each of our selected choices for~$\Phi$
and for six possible scenarios,
the table shows the largest value of~$q$
for which~$\Phi$ is $n$-cell implementable.
Due to the symmetries
$\E_q(x,t) = \E_q(t,x)$,
$\N_q(x,t) = \N_q(t,x)$, and
$\GE_q(x,t) = \GE_q(q{-}1{-}t,q{-}1{-}x) = \LE_q(t,x)$,
the columns~(\Sdb) and~(\Sbd) are the same
in the rows that correspond to
the families $\Eset_q$, $\Nset_q$, and $\GEset_q$,
and so are the columns~(\Srd) and~(\Sdr).
Hence, for these families, there are four scenarios to consider.

Three scenarios were omitted from the table:
Scenario~(\Sbb) corresponds to
\ifPAGELIMIT
    a BCAM and puts restrictions on
\else
a BCAM, where the two arguments of $\TCAM(\cdot,\cdot)$
are restricted to~$\Boolean$. Such a restriction limits
\fi
the subsets~$\Phi$ that can be expressed in the form~(\ref{eq:TCAM})
(regardless of how large~$n$
\ifPAGELIMIT
    is); among our families of interest, only $\Eset_q$ satisfies those
    restrictions (see Footnote~\ref{footnote:Eset} below).
    Scenarios~(\Sbr) and~(\Srb) are also
    omitted since---except for isolated cases---they coincide with
    Scenarios~(\Sbd) and~(\Sdb), respectively,
    in any of the families~$\Phi$ of interest.
\else
is), namely, if $f(y) = f(z) = 1$ for some $f \in \Phi$
and some two distinct $y, z \in [q \rangle$, then
$g(y) = g(z)$ for all $g \in \Phi$.
Among our families of interest, only $\Eset_q$ satisfies this
property (see Footnote~\ref{footnote:Eset} below).
Scenarios~(\Sbr) and~(\Srb) are also omitted.
Scenario~(\Sbr) is identical to~(\Sbd) unless~$\Phi$ contains
the all-zero function, and among our families of interest,
this occurs only when $\Phi = \Fset_q$;
in this case Scenarios~(\Sbr) and~(\Sbd) differ
only in one isolated
case (see Footnote~\ref{footnote:Fsetq=2} below).
Similarly, Scenarios~(\Srb) and~(\Sdb) would differ
only if there were an element $y \in [q \rangle$ such that
$f(y) = 0$ for all $f \in \Phi$; yet this does not happen
in any of the families~$\Phi$ of interest.
\fi

We see from the table that for a given $q \ge 3$
and under Scenarios~(\Sdb), (\Sbd), (\Sdr), and~(\Srr),
the family $\Nset_q$ is $n$-cell implementable, (if and) only if
the \emph{whole} set $\Fset_q$ is.
As such, the family $\Nset_q$, albeit forming an exponentially
small fraction of
\ifPAGELIMIT
    $\Fset_q$,
\else
$\Fset_q$ ($q$ out of the $2^q$ functions),
\fi
necessarily has the least-efficient CAM implementation
under these scenarios.
And this ``record'' is almost tied by the family
\ifPAGELIMIT
    $\GEset_q$.
\else
$\GEset_q$ (which, too, is of size~$q$).
\fi
On the other hand, under Scenarios~(\Sdd) and~(\Srd),
each of the families $\Nset_q$, $\GEset_q$
(or even $\GEset_q \cup \LEset_q$) is
$\lceil q/2 \rceil$-cell implementable,
while the whole set $\Fset_q$ requires (at least) $q$ cells
(i.e., for this set, Scenarios~(\Sdd) and~(\Srd)
do not offer any advantage over Scenarios~(\Sdb) or~(\Sbd)).

\begin{table*}[hbt]
\caption{Largest value of~$q$ that can be accommodated for
any number of cells~$n$.%
\ifPAGELIMIT
    \vspace{-2ex}
\fi}
\label{tab:summary}
\ifIEEE\else\scriptsize\fi
\[
\ifIEEE\else\arraycolsep0.9ex\fi
\ifPAGELIMIT
    \newcommand{\V}{\raisebox{0ex}[3.7ex][2.3ex]{}}
\else
\renewcommand{\arraystretch}{1.2}
\newcommand{\V}{\raisebox{0ex}[4.2ex][3.2ex]{}}
\fi
\begin{array}{|c||c|c|c|c|c|c|}
\hline
\V
\Phi
&
\makebox[19ex]{Scenario~(\Sdb)}
&
\makebox[19ex]{Scenario~(\Sbd)}
&
\makebox[19ex]{Scenario~(\Sdd)}
&
\makebox[19ex]{Scenario~(\Srd)}
&
\makebox[19ex]{Scenario~(\Sdr)}
&
\makebox[19ex]{Scenario~(\Srr)} \\
\hline\hline
\V
\Eset_q
& \multicolumn{3}{c|}{2^n}
& \multicolumn{2}{c|}{
\displaystyle
\binom{n}{\lfloor n/3\rfloor}\cdot 2^{\lceil 2n/3 \rceil}
}
&
\displaystyle\binom{2n}{n} \\
\hline
\V
\Nset_q
&
\multicolumn{2}{c|}{
  \arraycolsep2ex
  \begin{array}{ll}
  2 & \textrm{($n = 1$)} \\
  n & \textrm{($n \ge 2$)}
  \end{array}
}
& \multicolumn{4}{c|}{2n} \\
\hline
\V
\GEset_q
& \multicolumn{2}{c|}{n+1}
&
\arraycolsep1ex
\begin{array}{ll}
2n     & \textrm{($n = 1, 2$)} \\
2n{+}1 & \textrm{($n \ge 3$)}
\end{array}
& \multicolumn{3}{c|}{2n + 1} \\
\hline
\V
\GEset_q \cup \LEset_q
& n
& n + 1
& \multicolumn{2}{c|}{
  \arraycolsep2ex
  \begin{array}{ll}
  2    & \textrm{($n = 1$)} \\
  2n-1 & \textrm{($n \ge 2$)}
  \end{array}
}
& \multicolumn{2}{c|}{} \\
\cline{1-5}
\V
\Fset_q
& \multicolumn{4}{c|}{n}
& \multicolumn{2}{c|}{\raisebox{3.7ex}[0ex]{$2n$}} \\
\hline
\end{array}
\]%
\ifPAGELIMIT
    \vspace{-2ex}
\fi
\end{table*}

The problem of constructing efficient encodings
of integer intervals on a TCAM (which include in particular
efficient implementations of integer comparisons)
has been the studied in the literature, e.g., \cite{BHH},\cite{FRY}.
However, in those papers, the encoding of an interval indicator
(or of a comparison) can utilize several rows of the TCAM.
This translates into allowing taking a disjunction
(``or'', ``$\vee$'') of several terms of
the form seen in the right-hand side of~(\ref{eq:TCAM}),
with each term having its own word~$\bldtheta$
(the common mapping $x \mapsto \bldu(x)$ is usually taken
as the ordinary binary representation of~$x$).
However, in our setting, each TCAM row may contain
(the conjunction of) the implementation of several comparators
and, therefore, any such implementation cannot be split
across several rows.

\ifPAGELIMIT
\else
Our families of interest,
$\Eset_q$, $\Nset_q$, $\GEset_q$, $\GEset_q \cup \LEset_q$,
and $\Fset_q$ are covered, respectively,
in Sections~\ref{sec:Eset} through~\ref{sec:Fset}
(followed by a discussion in Section~\ref{sec:discussion}).
In each section except Section~\ref{sec:Fset},
the scenarios are covered in the order
they appear in Table~\ref{tab:summary}
(skipping Scenarios~(\Sdb) and~(\Srd) when
they coincide with Scenarios~(\Sbd) and~(\Sdr)).
Some effort has been put into making the text less tedious,
e.g., by combining several proofs.
\fi
We substantiate the values in Table~\ref{tab:summary}
by presenting both lower and upper bounds for each entry.
The lower bounds are all obtained by presenting explicit
mappings $x \mapsto \bldu(x)$
and $f \mapsto \bldtheta(f)$ for which Eq.~(\ref{eq:TCAM}) holds
for the respective family~$\Phi$.
Admittedly, many of the optimal
\ifPAGELIMIT
    mappings, such as those that are used in the (omitted) proof of
    Proposition~\ref{prop:simple1} below,
    are rather straightforward;
\else
mappings are rather straightforward
(e.g., they are either standard binary representations
of their argument or kinds of unary representations, as in the proofs of
Propositions~\ref{prop:simple1}--\ref{prop:simple3} below);
\fi
in those cases, the contribution of this work lies mainly
in showing that these mappings are, in fact, the best possible.
Notable exceptions are the mappings for
Scenarios~(\Srd), (\Sdr), and~(\Srr)
when $\Phi = \Eset_q$, and for Scenarios~(\Sdd) and~(\Srd)
when $\Phi = \GEset_q \cup \LEset_q$
(and, perhaps, to some extent, Scenario~(\Sdd) when $\Phi = \GEset_q$).
\ifPAGELIMIT
    Due to space limitations, we elaborate in subsequent sections
    only on selected entries in Table~\ref{tab:summary}
    (in particular, we will skip the entries that correspond to
    $\Nset_q$).
    \ifANONYMOUS
        Some omitted proofs can be found in the appendices.
    \else
        The omitted parts and proofs can be found in the full
        version of this work.
    \fi

    We end our introduction by stating the following simple
    observation, which leads to some of the lower bounds in
    Table~\ref{tab:summary}.

    \begin{proposition}
    \label{prop:simple1}
    \label{prop:simple3}
    Under each of the scenarios in Table~\ref{tab:summary},
    any subset $\Phi \subseteq \Fset_q$ is $n$-cell implementable,
    whenever $q \le n$.
    Moreover, under Scenario~\textup{(\Sdr)} or~\textup{(\Srr)},
    such a subset is $n$-cell implementable, whenever $q \le 2n$.
    \end{proposition}
\else
As part of our study of the whole set $\Fset_q$ in
Section~\ref{sec:Fset},
we also draw a connection between the subject of this paper
and that of computing the VC~dimension of Boolean monomials.

\subsection{Some simple observations}
\label{sec:simple}

We end our introduction by stating three simple observations
which will be useful in the sequel.

\begin{proposition}
\label{prop:simple1}
Under each of the scenarios in Table~\ref{tab:summary},
any subset $\Phi \subseteq \Fset_q$ is $n$-cell implementable, whenever
\[
q \le n .
\]
\end{proposition}

\begin{proof}
Given that $q \le n$, we can extend (arbitrarily)
the functions in $\Fset_q$ so that they are defined
over the (possibly larger) domain
$[n \rangle$ and, so, we assume that $q = n$.
For Scenario~(\Sbd), we take the mappings
$x \mapsto \bldu(x) = (u_j(x))_{j \in [n \rangle}$ and
$f \mapsto \bldtheta(f) = (\vartheta_j(f))_{j \in [n \rangle}$ to be
\[
u_j(x) =
\left\{
\arraycolsep0.8ex
\begin{array}{ll}
1   & \textrm{if $x = j$}   \\
0   & \textrm{if $x \ne j$}
\end{array}
\right.
\quad \textrm{and} \quad
\vartheta_j(f)
=
\left\{
\arraycolsep0.8ex
\begin{array}{ll}
{*} & \textrm{if $f(j) = 1$} \\
0   & \textrm{if $f(j) = 0$}
\end{array}
\right.
.
\]
Otherwise, we take them to be
\[
u_j(x) =
\left\{
\begin{array}{ll}
1   & \textrm{if $x = j$} \\
{*} & \textrm{if $x \ne j$}
\end{array}
\right.
\quad \textrm{and} \quad
\vartheta_j(f) = f(j) .
\]
\end{proof}

Taking the family
$\Phi = \Nset_q = \{ \N_q(\cdot, t) \}_{t \in [q \rangle}$
as an example, Table~\ref{tab:NEsetSbd} shows the respective
mappings $x \mapsto \bldu(x)$
and $\N_q(\cdot,t) \mapsto \bldtheta(\N_q(\cdot,t))$
under Scenario~(\Sbd).

\begin{table}[hbt]
\caption{Mappings for $\Phi = \Nset_q$ that attain $q = n$
under Scenario~(\Sbd).}
\label{tab:NEsetSbd}
\[
\begin{array}{c|c}
x      & \bldu(x)                      \\
\hline
\raisebox{0ex}[2ex]{}
0      & 10000             \ldots 0    \\
1      & 01000             \ldots 0    \\
2      & 00100             \ldots 0    \\
\vdots & \hspace{2ex}\ddots            \\
n{-}2  & 000              \ldots 010   \\
n{-}1  & 0000             \ldots 01
\end{array}
\quad\quad\quad
\begin{array}{c|c}
t      & \bldtheta(\N_q(\cdot,t))      \\
\hline
\raisebox{0ex}[2ex]{}
0      & 0{*}{*}{*}{*} \ldots {*}      \\
1      & {*}0{*}{*}{*} \ldots {*}      \\
2      & {*}{*}0{*}{*} \ldots {*}      \\
\vdots & \hspace{1ex} \ddots           \\
n{-}2  & {*}{*}{*}     \ldots {*}0{*} \\
n{-}1  & {*}{*}{*}{*}  \ldots {*}0    \\
\end{array}
\]
\end{table}

\begin{proposition}
\label{prop:simple2}
Under each of the scenarios in Table~\ref{tab:summary},
any subset $\Phi \subseteq \Fset_q$ is $n$-cell implementable, whenever
\[
|\Phi| \le n .
\]
\end{proposition}

\begin{proof}
Again, we can assume that $|\Phi| = n$.
Writing $\Phi = \{ f_1, f_2, \ldots, f_{|\Phi|} \}$,
for Scenario~(\Sdb), we take the mappings
$x \mapsto \bldu(x)$ and $f \mapsto \bldtheta(f)$
to be
\[
u_j(x) =
\left\{
\arraycolsep0.8ex
\begin{array}{ll}
{*} & \textrm{if $f_j(x) = 1$} \\
0   & \textrm{if $f_j(x) = 0$}
\end{array}
\right.
\;\; \textrm{and} \;\;
\vartheta_j(f)
=
\left\{
\arraycolsep0.8ex
\begin{array}{ll}
1   & \textrm{if $f =   f_j$} \\
0   & \textrm{if $f \ne f_j$}
\end{array}
\right.
.
\]
Otherwise, we take them to be
\[
u_j(x) = f_j(x)
\quad \textrm{and} \quad
\vartheta_j(f)
=
\left\{
\begin{array}{ll}
1   & \textrm{if $f   = f_j$} \\
{*} & \textrm{if $f \ne f_j$}
\end{array}
\right.
.
\]
\end{proof}

\begin{proposition}
\label{prop:simple3}
Under each of the scenarios~\textup{(\Sdr)} or~\textup{(\Srr)},
any subset $\Phi \subseteq \Fset_q$ is
$n$-cell implementable, whenever
\[
q \le 2n .
\]
\end{proposition}

\begin{proof}
Assuming that $q = 2n$, we take the mappings
$x \mapsto \bldu(x)$ and $f \mapsto \bldtheta(f)$ to be
\begin{equation}
\label{eq:simple3u}
u_j(x) =
\left\{
\begin{array}{lcl}
0   && \textrm{if $x = 2j$}   \\
1   && \textrm{if $x = 2j+1$} \\
{*} && \textrm{otherwise}
\end{array}
\right.
\phantom{.}
\end{equation}
and
\begin{equation}
\label{eq:simple3theta}
\vartheta_j(f) =
\left\{
\begin{array}{lcl}
{*}     && \textrm{if $f(2j) = f(2j{+}1) = 1$}         \\
0       && \textrm{if $f(2j) = 1$ and $f(2j{+}1) = 0$} \\
1       && \textrm{if $f(2j) = 0$ and $f(2j{+}1) = 1$} \\
\reject && \textrm{if $f(2j) =  f(2j{+}1) = 0$}
\end{array}
\right.
.
\end{equation}
\end{proof}
\fi

\section{The family $\Eset_q$}
\label{sec:Eset}

In this section, we consider the family $\Phi = \Eset_q$.
Since each function $x \mapsto \E_q(x,t) \in \Eset_q$
can be identified by the parameter $t \in [q \rangle$,
we will use for convenience the notation $t \mapsto \bldtheta(t)$
for $\E_q(\cdot,t) \mapsto \bldtheta(\E_q(\cdot,t))$.
Due to the symmetry $\E_q(x,t) = \E_q(t,x)$,
it will suffice to state our results only for
Scenarios~(\Sbd), (\Sdd), (\Sdr), and~(\Srr).
\ifPAGELIMIT
\else

The next proposition handles the case $\Phi = \Eset_q$
under Scenarios~(\Sbd) and~(\Sdd).
\fi

\begin{proposition}
\label{prop:EsetSbdSdd}
Under each of the scenarios~\textup{(\Sbd)} or~\textup{(\Sdd)},
the family $\Eset_q$ is $n$-cell implementable, \ifandonlyif\
\ifPAGELIMIT
    $q \le 2^n$.
\else
\[
q \le 2^n .
\]
\fi
\end{proposition}

\ifPAGELIMIT
    Sufficiency follows by letting $\bldu(x)$ and $\bldtheta(t)$ be
    the length-$n$ binary representations of~$x$ and~$t$,
    respectively.\footnote{%
    \label{footnote:Eset}
    Such mappings suit also Scenario~(\Sbb).}
\else
\begin{proof}
Sufficiency follows by letting $\bldu(x)$ and $\bldtheta(t)$ be
the length-$n$ binary representations of~$x$ and~$t$,
respectively.\footnote{%
\label{footnote:Eset}
Such mappings suit also Scenario~(\Sbb).}

Necessity under Scenario~(\Sbd) follows from the fact
that the mapping $x \mapsto \bldu(x)$ must be injective
(otherwise, if $\bldu(x_0) = \bldu(x_1)$ for $x_0 \ne x_1$,
then Eq.~(\ref{eq:TCAM})
would imply $\E_q(x_0,x_1) = \E_q(x_1,x_1) = 1$, which is
a contradiction).

As for Scenario~(\Sdd), suppose that Eq.~(\ref{eq:TCAM}) holds
for some $x \mapsto \bldu(x)$ and $t \mapsto \bldtheta(t)$.
Define the mapping
$x \mapsto \widehat{\bldu}(x) = (\widehat{u}_j(x))_{j \in [n \rangle}$
as follows:
\[
\widehat{u}_j(x) =
\left\{
\begin{array}{ccl}
u_j(x)         && \textrm{if $u_j(x) \in \Boolean$} \\
\vartheta_j(x) &&
         \textrm{if $u_j(x) = {*}$ and $\vartheta_j(x) \in \Boolean$} \\
0              && \textrm{if $u_j(x) = \vartheta_j(x) = {*}$}
\end{array}
\right.
.
\]
It readily follows that
$\TCAM(\widehat{\bldu}(x),\bldtheta(x))
= \TCAM(\bldu(x),\bldtheta(x)) = 1$
for every $x \in [q \rangle$. On the other hand, since
$\widehat{\bldu}(x)$ is obtained from $\bldu(x)$ by
(possibly) changing some entries from~$*$ into elements of~$\Boolean$,
then $\TCAM(\bldu(x),\bldtheta(t)) = 0$ implies
$\TCAM(\widehat{\bldu}(x),\bldtheta(t)) = 0$.
Hence, $\E_q(x,t) = \TCAM(\widehat{\bldu}(x),\bldtheta(t))$
for all $x, t \in [q \rangle$, thereby reducing to Scenario~(\Sbd).
\end{proof}
\fi

The family $\Phi = \Eset_q$
under Scenarios~(\Sdr) and~(\Srr) will be treated next,
yet, to this end, we will need some definitions
and two lemmas.

We introduce the following partial ordering, $\preceq$, on
the elements of $\Boolean_\reject$:
\[
\reject \;\;\preceq\;\; 0, 1  \;\;\preceq\;\; {*}
\]
(with no ordering defined between~$0$ and~$1$),
and extend it to words $\bldv, \bldv' \in \Boolean_\reject^n$,
with $\bldv \preceq \bldv'$ if the relation holds componentwise.

A \emph{(maximal) chain} over $\Boolean_*^n$ is a list of $n+1$ words
in $\Boolean_*^n$,
\ifPAGELIMIT
    $(\bldv_0, \bldv_1, \ldots, \bldv_n)$,
\else
\[
(\bldv_0, \bldv_1, \ldots, \bldv_n) ,
\]
\fi
where $\bldv_0$ is a word in $\Boolean^n$
and for each $i \in [n \rangle$,
the word $\bldv_{i+1}$ is obtained from $\bldv_i$ by changing one
non-$*$ entry into a~$*$; thus, $\bldv_n$ is always the all-$*$
\ifPAGELIMIT
    word.
\else
word.\footnote{The definition of a chain
can also be extended to $\Boolean_\reject^n$,
but our analysis below will use a shortcut that will not require
this extension.}
\fi
It is rather easy to see that
the number of chains over $\Boolean_*^n$ is $2^n \cdot n!$.

An \emph{antichain} in $\Boolean_*^n$
(respectively, $\Boolean_\reject^n$) is
a subset~$\Antichain$ of $\Boolean_*^n$
(respectively, $\Boolean_\reject^n$)
such that $\bldv \npreceq \bldv'$ for any two distinct words
$\bldv, \bldv' \in \Antichain$.

For a word $\bldv \in \Boolean_\reject^n$
and a symbol $\sigma \in \Boolean_\reject$,
we denote by $\weight_\sigma(\bldv)$ the number of occurrences
of the symbol~$\sigma$ in~$\bldv$.
For $w \in [0:n]$, let
\[
\Boolean_*^n(w) =
\bigl\{ \bldv \in \Boolean_*^n \,:\,
\weight_0(\bldv) + \weight_1(\bldv) = w \bigr\} .
\]

The following lemma is a generalization
of Sperner's theorem~\cite[p.~2]{Anderson}.

\begin{lemma}
\label{lem:SpernerSdr}
Let $\Antichain$ be an antichain in $\Boolean_*^n$. Then
\[
|\Antichain| \le
\binom{n}{\lfloor n/3 \rfloor} \cdot 2^{\lceil 2n/3 \rceil} ,
\]
with equality attained by the set
$\Boolean_*^n (\lceil 2n/3 \rceil)$.
\end{lemma}

\begin{proof}
We adapt Lubell's proof of Sperner's theorem (in~\cite{Lubell})
to our setting.
Let~$\Antichain$ be an antichain and
for $w \in [0:n]$, let
$\Antichain(w) = \Antichain \cap \Boolean_*^n(w)$.
We first observe that any word $\bldv \in \Antichain(w)$
is contained in exactly
\[
2^{n-w} \cdot w! (n-w)!
\]
chains over $\Boolean_*^n$
and none of the other words in those chains belongs to $\Antichain$.
Hence,
\[
\sum_{w \in [0:n]}
2^{n-w} \cdot w! (n-w)! \cdot |\Antichain(w)| \le 2^n \cdot n! ,
\]
which can also be written as
\begin{equation}
\label{eq:LYM}
\sum_{w \in [0:n]}
\frac{|\Antichain(w)|}{\beta(w)} \le 1 ,
\end{equation}
where
\[
\beta(w) = 2^w \cdot \binom{n}{w} .
\]
The latter expression is maximized
when $w = w_{\max} = \lceil 2n/3 \rceil$ and, so,
\ifPAGELIMIT
    \begin{eqnarray*}
    |\Antichain|
    & = &
    \sum_{w \in [0:n]} |\Antichain(w)| 
    \; \le \;
    \sum_{w \in [0:n]}
    \frac{\beta(w_{\max})}{\beta(w)} \cdot |\Antichain(w)| \\
    & \stackrel{\textrm{(\ref{eq:LYM})}}{\le} &
    \beta(w_{\max}) =
    \binom{n}{\lfloor n/3 \rfloor} \cdot 2^{\lceil 2n/3 \rceil} .
    \end{eqnarray*}
\else
\begin{eqnarray*}
|\Antichain|
& = &
\sum_{w \in [0:n]} |\Antichain(w)| \\
& \le &
\sum_{w \in [0:n]}
\frac{\beta(w_{\max})}{\beta(w)} \cdot |\Antichain(w)| \\
& \stackrel{\textrm{(\ref{eq:LYM})}}{\le} &
\beta(w_{\max}) =
\binom{n}{\lfloor n/3 \rfloor} \cdot 2^{\lceil 2n/3 \rceil} .
\end{eqnarray*}
\fi
The inequalities become equalities when
$\Antichain = \Boolean_*^n( \lceil 2n/3 \rceil)$
which is, indeed, an antichain.
\end{proof}

The counterpart of Lemma~\ref{lem:SpernerSdr}
for $\Boolean_\reject^n$ takes the following form.

\begin{lemma}
\label{lem:SpernerSrr}
Let $\Antichain$ be an antichain in $\Boolean_\reject^n$. Then
\[
|\Antichain| \le \binom{2n}{n} ,
\]
with equality attained by the set
\[
\Boolean_\circledast^n =
\bigl\{
\bldv \in \Boolean_\reject^n \,:\,
\weight_\reject(\bldv) = \weight_*(\bldv)
\bigr\}
\ifPAGELIMIT
    .
\fi
\]
\ifPAGELIMIT
\else
(namely, $\Boolean_\circledast^n$ is
the set of all words in $\Boolean_\reject^n$ in which
the symbols~$\reject$ and $*$ have the same count).
\fi
\end{lemma}

\ifPAGELIMIT
\else
\begin{proof}
Consider the following mapping
$\lambda : \Boolean_\reject \rightarrow \Boolean^2$:
\[
\lambda(\reject) = 00 , \quad
\lambda(0) = 01 , \quad
\lambda(1) = 10 , \quad
\lambda({*}) = 11 .
\]
We extend this definition to a mapping from words
$\bldv = (v_j)_{j \in [n \rangle} \in \Boolean_\reject^n$
to words in $\Boolean^{2n}$ by
\[
\lambda(\bldv) =
\lambda(v_0) \|
\lambda(v_1) \|
\ldots \|
\lambda(v_{n-1})
\]
(with $\|$ denoting concatenation)
and, accordingly, from subsets of $\Boolean_\reject^n$
to subsets of $\Boolean^{2n}$. Clearly, $\lambda$
is a bijection under all these settings.
Moreover, a subset $\Antichain \subseteq \Boolean_\reject^n$
is an antichain in $\Boolean_\reject^n$, \ifandonlyif\
$\lambda(\Antichain)$ is an antichain in $\Boolean^{2n}$
under the ordering $0 \le 1$ on the elements of $\Boolean$;
namely, there are no two distinct words $\bldv, \bldv' \in \Antichain$
such that $\lambda(\bldv) \le \lambda(\bldv')$ (componentwise).
Hence, by Sperner's theorem we have,
for every antichain~$\Antichain$ in $\Boolean_\reject^n$,
\[
|\Antichain| = |\lambda(\Antichain)| \le \binom{2n}{n} ,
\]
with equality holding when~$\Antichain$ is the following
subset of $\Boolean_\reject^n$ (see~\cite[p.~2]{Anderson}):
\[
\bigl\{
\bldv \in \Boolean_\reject^n \,:\,
\weight_0(\lambda(\bldv)) = \weight_1(\lambda(\bldv)) = n \bigr\} .
\]
One can easily verify that this set is identical
to~$\Boolean_\circledast^n$.
\end{proof}

We now turn to stating our result
for the family $\Phi = \Eset_q$ under
Scenarios~(\Sdr) and~(\Srr).
\fi

\begin{proposition}
\label{prop:EsetSdrSrr}
The following holds for the subset $\Eset_q$.
\begin{list}{}{\settowidth{\labelwidth}{\textit{(a)}}}
\item[(a)]
Under Scenario~\textup{(\Sdr)},
it is $n$-cell implementable, \ifandonlyif
\[
q \le \binom{n}{\lfloor n/3 \rfloor} \cdot 2^{\lceil 2n/3 \rceil} .
\]
\item[(b)]
Under Scenario~\textup{(\Srr)},
it is $n$-cell implementable, \ifandonlyif
\[
q \le \binom{2n}{n} .
\]
\end{list}
\end{proposition}

\begin{proof}
Starting with proving necessity, we observe that the mapping
$x \mapsto \bldu(x)$ has to be injective and its set of images
has to be an antichain in $\Boolean_*^n$ (for part~(a))
or in $\Boolean_\reject^n$
\ifPAGELIMIT
    (for part~(b)).
\else
(for part~(b)):
indeed, if we had $\bldu(x) \preceq \bldu(x')$
for distinct $x, x' \in [q \rangle$ then, from
$\E_q(x,x) = \TCAM(\bldu(x),\bldtheta(x)) = 1$
we would get the contradiction
$\E_q(x',x) = \TCAM(\bldu(x'),\bldtheta(x)) = 1$.
\fi
The sought result then follows from
Lemmas~\ref{lem:SpernerSdr} and~\ref{lem:SpernerSrr}.

Sufficiency follows by selecting
$x \mapsto \bldu(x)$ to be
any injective mapping into
$\Boolean_*^n( \lceil 2n/3 \rceil)$ (for part~(a))
or into $\Boolean_\circledast^n$ (for part~(b))
and defining
$t \mapsto \bldtheta(t)$
for every $t \in [q \rangle$ as follows:
\[
\vartheta_j(t)
=
\left\{
\begin{array}{ccl}
u_j(t)  && \textrm{if $u_j(t) \in \Boolean$} \\
\reject && \textrm{if $u_j(t) = {*}$} \\
{*}     && \textrm{if $u_j(t) = \reject$}
\end{array}
\right.
.
\]
\end{proof}

Using known approximations for
the binomial coefficients (see~\cite[p.~309]{MS}),
the inequalities in parts~(a) and~(b) of
Proposition~\ref{prop:EsetSdrSrr} are implied by
\begin{equation}
\ifPAGELIMIT
    \nonumber
\else
    \ifARXIV
        \nonumber
     \fi
\fi
q \le \frac{3}{4 \sqrt{n}} \cdot 3^n
\quad \textrm{and} \quad
q \le  \frac{1}{2 \sqrt{n}} \cdot 4^n ,
\end{equation}
respectively.
This means that~$q$ can almost get
to the size of $\Boolean_*^n$ (in part~(a)) or of $\Boolean_\reject^n$
\ifPAGELIMIT
    (in part~(b)).
\else
(in part~(b))
(this size would have been reachable if $\TCAM(\cdot,\cdot)$
had been defined so that $\TCAM(u,\vartheta) = 1$ \ifandonlyif\
$u = \vartheta$, for any $u, \vartheta \in \Boolean_\reject$).

\begin{remark}
\label{rem:lambda}
Referring to the mapping
$\lambda : \Boolean_\reject \rightarrow \Boolean^2$
in the proof of Lemma~\ref{lem:SpernerSrr},
for any $u \in \Boolean_\reject$,
denote by $\lambda_0(u)$ and $\lambda_1(u)$
the first and second entries (in $\Boolean$)
of $\lambda(u)$. It is fairly easy
to see that for every $u, \vartheta \in \Boolean_\reject$,
\[
\TCAM(u,\vartheta)
= 
\bigl( \lambda_0(u) \vee \lambda_1(\vartheta) \bigr)
\wedge
\bigl( \lambda_1(u) \vee \lambda_0(\vartheta) \bigr) .
\]
In fact, this is (essentially) how TCAM cells
are implemented in~\cite[Fig.~2(a)]{GG} and in~\cite[Fig.~1]{KTFY}.
Given a family $\Phi \subseteq \Fset_q$,
we can therefore rewrite~(\ref{eq:TCAM}) as
\[
f(x) = 
\bigwedge_{j \in [2n \rangle}
\left( \widetilde{u}_j(x) \vee \widetilde{\vartheta}_j(f) \right)
\]
with mappings
$\widetilde{\bldu} : [q \rangle \rightarrow \Boolean^{2n}$
and
$\widetilde{\bldtheta} : \Phi \rightarrow \Boolean^{2n}$
that satisfy certain constraints, depending on the scenario
(the only exception is Scenario~(\Srr), where no constraints
are imposed).\qed
\end{remark}

\section{The family $\Nset_q$}
\label{sec:Nset}

In this section, we consider the family $\Phi = \Nset_q$.
As we did in Section~\ref{sec:Eset},
we will identify each function $x \mapsto \N_q(x,t) \in \Nset_q$
by the parameter $t \in [q \rangle$
and will use the notation $t \mapsto \bldtheta(t)$
for $\N_q(\cdot,t) \mapsto \bldtheta(\N_q(\cdot,t))$.
And, here, too, it will suffice to state the results only for
Scenarios~(\Sbd), (\Sdd), (\Sdr), and~(\Srr).

\begin{proposition}
\label{prop:NsetSbd}
Under Scenario~\textup{(\Sbd)},
the family $\Nset_q$ is $n$-cell implementable, \ifandonlyif
\[
q \le
\left\{
\begin{array}{lcl}
2 && \textrm{if $n = 1$} \\
n && \textrm{if $n \ge 2$}
\end{array}
\right.
.
\]
\end{proposition}

\begin{proof}
The case $n = 1$ is easily verified, so we assume
hereafter in the proof that $n \ge 2$.

Sufficiency follows from Proposition~\ref{prop:simple1}
or~\ref{prop:simple2} (see Table~\ref{tab:NEsetSbd}),
and necessity follows essentially from
the proof of Theorem~1 in~\cite{NS};
we include a proof for completeness and for reference in the sequel.
Suppose that Eq.~(\ref{eq:TCAM}) holds for~$\Nset_q$
with mappings $\bldu : [q \rangle \rightarrow \Boolean^n$ and
$\bldtheta : \Nset_q \rightarrow \Boolean_*^n$;
both these mappings are necessarily injective.
For every $t \in [q \rangle$
we have $\TCAM(\bldu(t),\bldtheta(t)) = \N_q(t,t) = 0$
and, so, there exists an index $j \in [n \rangle$
such that $\vartheta_j(t) \in \Boolean$ and
$u_j(t) \ne \vartheta_j(t)$; moreover, there is no loss of generality
in assuming that all the other entries in $\bldtheta(t)$ are~$*$.
Assume now to the contrary that $q > n \ge 2$.
By the pigeonhole principle there exist
distinct $t_0, t_1 \in [q \rangle$ for which that index~$j$ is the same.
Since the mapping~$\bldtheta$ is injective,
we must have $\vartheta_j(t_0) \ne \vartheta_j(t_1)$.
Yet then, for (any) $t_2 \ne t_0, t_1$ in $[q \rangle$
we have $u_j(t_2) \ne \vartheta_j(t_i)$ for some $i \in [2 \rangle$,
thereby yielding the contradiction
$\N_q(t_2,t_i) = \TCAM(\bldu(t_2),\bldtheta(t_i)) = 0$.
\end{proof}

Recall from Propositions~\ref{prop:simple1} and~\ref{prop:simple2} that
the whole set $\Fset_q$ is $q$-cell implementable
and that any subset~$\Phi$ of $\Fset_q$
of size $|\Phi| < q$ is $(q-1)$-cell implementable.
Proposition~\ref{prop:NsetSbd} implies that
for any $q \ge 3$ and under Scenarios~(\Sdb) and~(\Sbd),
the family $\Nset_q$ is a smallest possible
subset of $\Fset_q$ that requires the same number of cells, $q$,
as the whole set $\Fset_q$ does.

\begin{proposition}
\label{prop:NsetSdd}
Under each of
the scenarios~\textup{(\Sdd)}, \textup{(\Sdr)}, or~\textup{(\Srr)},
the family $\Nset_q$ is $n$-cell implementable, \ifandonlyif
\[
q \le 2n .
\]
\end{proposition}

\begin{proof}
Sufficiency follows from Proposition~\ref{prop:simple3}
which, for the family $\Nset_q$, holds in fact
also under Scenario~(\Sdd): specifically, for this family,
the range of the mapping~$\bldtheta$
in~(\ref{eq:simple3theta}) is actually $\Boolean_*^n$.

We next prove necessity under (the loosest) Scenario~(\Srr), and
our arguments will be similar to those used in the previous proof.
Specifically, assuming that Eq.~(\ref{eq:TCAM}) holds,
for every $t \in [q \rangle$
there exists an index $j \in [n \rangle$
such that $u_j(t), \vartheta_j(t) \in \Boolean \cup \{ \reject \}$ and
either $\vartheta_j(t) = \reject$ or $u_j(t) \ne \vartheta_j(t)$;
moreover, all the other entries in $\bldtheta(t)$ can be assumed
to be~$*$. Suppose to the contrary that $q > 2n$.
By the pigeonhole principle there exist
distinct $t_0, t_1, t_2 \in [q \rangle$
for which that index~$j$ is the same.
Since~$\bldtheta$ is injective,
we must have $\vartheta_j(t_i) = \reject$
for (exactly) one $i \in [3 \rangle$.
Yet then $\N_q(t_k,t_i) = \TCAM(\bldu(t_k),\bldtheta(t_i)) = 0$
for $k \in [3 \rangle \setminus \{ i \}$, which is a contradiction.
\end{proof}
\fi

\section{The family $\GEset_q$}
\label{sec:GEset}

In this section, we consider the family $\Phi = \GEset_q$.
We again identify each function $x \mapsto \GE_q(x,t) \in \GEset_q$
by the parameter $t \in [q \rangle$
and use the notation
$t \mapsto \bldtheta(t)$
for $\GE_q(\cdot,t) \mapsto \bldtheta(\GE_q(\cdot,t))$.
Clearly, $\GE_q(x,t) = \TCAM(\bldu(x),\bldtheta(t))$
\ifandonlyif\
$\LE_q(x,t) = \TCAM(\bldu'(x),\bldtheta'(t))$,
where $\bldu'(x) = \bldu(q{-}1{-}x)$
and $\bldtheta'(t) = \bldtheta(q{-}1{-}t)$;
hence, our analysis will apply just as well to $\Phi = \LEset_q$.
Moreover, since $\GE_q(x,t) = \LE_q(t,x)$,
it will suffice to state the results only for
Scenarios~(\Sbd), (\Sdd), (\Sdr), and~(\Srr).

\begin{proposition}
\label{prop:GEsetSbd}
Under Scenario~\textup{(\Sbd)},
the family $\GEset_q$ is $n$-cell implementable, \ifandonlyif\
\ifPAGELIMIT
    $q \le n + 1$.
\else
\[
q \le n + 1.
\]
\fi
\end{proposition}

\begin{proof}
Necessity will follow from the proof of
Proposition~\ref{prop:GEsetSdd} below and
sufficiency follows by defining
$x \mapsto \bldu(x)$ and $t \mapsto \bldtheta(t)$ as follows:
\[
u_j(x) =
\left\{
\begin{array}{ll}
0   & \textrm{if $x \le j$} \\
1   & \textrm{if $x >   j$}
\end{array}
\right.
\quad \textrm{and} \quad
\vartheta_j(t) =
\left\{
\begin{array}{ll}
{*} & \textrm{if $t \le j$} \\
1   & \textrm{if $t > j$}
\end{array}
\right.
\ifPAGELIMIT
    .
\fi
\]
\ifPAGELIMIT
\else
(these mappings are also shown in Table~\ref{tab:GEsetSbd}).
\fi
\end{proof}

\ifPAGELIMIT
\else
\begin{table}[hbt]
\caption{Mappings for $\GEset_q$ that attain $q = n+1$
under Scenario~(\Sbd).}
\label{tab:GEsetSbd}
\[
\begin{array}{c|c}
x      & \bldu(x)                   \\
\hline
\raisebox{0ex}[2ex]{}
0      & 00000             \ldots 0 \\
1      & 10000             \ldots 0 \\
2      & 11000             \ldots 0 \\
3      & 11100             \ldots 0 \\
\vdots & \hspace{2ex}\ddots         \\
n{-}1  & 1111             \ldots 10 \\
n      & 11111             \ldots 1
\end{array}
\quad\quad\quad
\begin{array}{c|c}
t      & \bldtheta(t)               \\
\hline
\raisebox{0ex}[2ex]{}
0      & {*}{*}{*}{*}{*} \ldots {*} \\
1      & 1{*}{*}{*}{*}   \ldots {*} \\
2      & 11{*}{*}{*}     \ldots {*} \\
3      & 111{*}{*}       \ldots {*} \\
\vdots & \hspace{1ex} \ddots        \\
n{-}1  & 1111           \ldots 1{*} \\
n      & 11111             \ldots 1
\end{array}
\]
\end{table}
\fi

\begin{proposition}
\label{prop:GEsetSdd}
Under Scenario~\textup{(\Sdd)},
the family $\GEset_q$ is $n$-cell implementable, \ifandonlyif
\[
q \le
\left\{
\begin{array}{lcl}
2n     && \textrm{if $n = 1, 2$} \\
2n + 1 && \textrm{if $n \ge 3$}
\end{array}
\right.
.
\]
\end{proposition}

\begin{proof}
Starting with proving necessity,
suppose that Eq.~(\ref{eq:TCAM}) holds.
For each $j \in [n \rangle$, let $y_j$ denote the smallest
$t \in [q \rangle$ such that $\vartheta_j(t) \in \Boolean$
(define $y_j = \infty$ if no such~$t$
\ifPAGELIMIT
    exists); without loss of generality we can assume
\else
exists). By possibly switching between the roles of~$0$ and~$1$
in $t \mapsto \vartheta_j(t)$ and in $x \mapsto u_j(x)$,
we can assume without loss of generality
\fi
that $\vartheta_j(y_j) = 1$.
We now observe from Eq.~(\ref{eq:TCAM}) that
$u_j(x) \in \{ 1, * \}$ for every $x \ge y_j$.
Thus, (\ref{eq:TCAM}) still holds if
we re-define $\vartheta_j(t)$ to be equal to~$1$
at every $t > y_j$ for which $\vartheta_j(t) = {*}$.
Hence, we assume hereafter without loss of generality that
\[
\vartheta_j(y_j) = 1 \;\; \textrm{and} \;\;
\vartheta_j(t) \in \Boolean \;\; \textrm{for every $t > y_j$}
\]
(provided that $y_j < \infty$).

Next, for each $j \in [n \rangle$, we let $z_j$ be the smallest
$t \in [y_j:q \rangle$ such that $\vartheta_j(t) = 0$
(with $z_j = \infty$ if no such~$t$ exists).
Note that for Eq.~(\ref{eq:TCAM}) to hold, we must have
\begin{equation}
\label{eq:x}
u_j(x) = {*} \;\; \textrm{for every $x \ge z_j$}.
\end{equation}
In particular, under Scenario~(\Sbd)
(as in Proposition~\ref{prop:GEsetSbd}),
we must have $z_j = \infty$ for every $j \in [n \rangle$.

In summary,
\ifPAGELIMIT
    we can assume that
\else
the mapping
$t \mapsto \bldtheta(t)$
is assumed hereafter to take the following form:
\fi
\begin{equation}
\label{eq:t}
\vartheta_j(t) =
\left\{
\begin{array}{lcl}
{*}                 && \textrm{if $t < y_j$}      \\
1                   && \textrm{if $y_j\le t<z_j$} \\
0                   && \textrm{if $t = z_j$}      \\
0 \;\textrm{or}\; 1 && \textrm{if $t > z_j$}
\end{array}
\right.
.
\end{equation}

Next, we claim that
\begin{equation}
\label{eq:GEsetSdd}
\ifPAGELIMIT
    \textstyle
\fi
[1:q \rangle \subseteq
\bigcup_{j \in [n \rangle} \{ y_j, z_j  \}
\end{equation}
(where the right-hand side is regarded as a set, ignoring
multiplicities). Indeed, suppose to the contrary
that there exists some $y \in [1:q \rangle$ such that
$y \notin \{ y_j, z_j \}$ for every $j \in [n \rangle$.
By~(\ref{eq:x}) and~(\ref{eq:t})
we then have, for every $j \in [n \rangle$:
\begin{equation}
\label{eq:yjzj}
\begin{array}{lcl}
\vartheta_j(y) = {*}                     && \textrm{if $y < y_j$}   \\
\vartheta_j(y)=\vartheta_j(y-1) \; (= 1) && \textrm{if $y_j<y<z_j$} \\
u_j(y-1) = {*}                           && \textrm{if $y > z_j$}
\end{array}
.
\end{equation}
\ifPAGELIMIT
    Hence,
    \begin{eqnarray*}
    \lefteqn{
    \GE_q(y-1,y)
    \stackrel{\textrm{(\ref{eq:TCAM})+(\ref{eq:yjzj})}}{=}
    \textstyle
    \bigwedge_{j \,:\, y_j < y < z_j}
    \TCAM(u_j(y-1),\vartheta_j(y))} \makebox[5ex]{} \\
    & \stackrel{\textrm{(\ref{eq:yjzj})}}{=} &
    \textstyle
    \bigwedge_{j \,:\, y_j < y < z_j}
    \TCAM(u_j(y-1),\vartheta_j(y-1))
    \; \stackrel{\textrm{(\ref{eq:TCAM})}}{=} \; 1,
\end{eqnarray*}
\else
Hence,\footnote{A conjunction over an empty set is defined to be~$1$.}
\begin{eqnarray*}
\GE_q(y-1,y)
& \stackrel{\textrm{(\ref{eq:TCAM})}}{=} &
\TCAM(\bldu(y-1),\bldtheta(y)) \\
& \stackrel{y \ne y_j, z_j}{=} &
\Bigl(
\bigwedge_{j \,:\, y < y_j}
\TCAM(u_j(y-1),{\underbrace{\vartheta_j(y)}_{*}})
\Bigr) \\
&&
\quad {}
\wedge
\Bigl( \bigwedge_{j \,:\, y_j < y < z_j}
\!\!\!\!
\TCAM(u_j(y-1),\vartheta_j(y)) \Bigr) \\
&&
\quad {}
\wedge
\Bigl( \bigwedge_{j \,:\, y > z_j}
\TCAM({\underbrace{u_j(y-1)}_{*}},\vartheta_j(y)) \Bigr) \\
& \stackrel{\textrm{(\ref{eq:yjzj})}}{=} &
\bigwedge_{j \,:\, y_j < y < z_j}
\TCAM(u_j(y-1),\vartheta_j(y)) \\
& \stackrel{\textrm{(\ref{eq:yjzj})}}{=} &
\bigwedge_{j \,:\, y_j < y < z_j}
\TCAM(u_j(y-1),\vartheta_j(y-1)) \\
& \stackrel{\textrm{(\ref{eq:TCAM})}}{=} &
1,
\end{eqnarray*}
\fi
which is a contradiction.
By~(\ref{eq:GEsetSdd}) we thus conclude that $q-1 \le 2n$,
thereby proving the necessary condition for $n \ge 3$
\ifPAGELIMIT
    (the special cases $n = 1, 2$ can be easily verified separately).
\else
(leaving the special cases of $n = 1, 2$ to
Appendix~\ref{sec:specialcases}).
\fi
Moreover, when the range of~$\bldu$ is constrained to $\Boolean^n$,
then $z_j = \infty$ for every $j \in [n \rangle$ and, so,
by~(\ref{eq:GEsetSdd}) we get $q-1 \le n$,
thus proving the necessary condition
in Proposition~\ref{prop:GEsetSbd}.

Sufficiency follows from the mappings shown
\ifPAGELIMIT
    in Table~\ref{tab:GEsetSdd}.
\else
in Table~\ref{tab:GEsetSdd};
in that table, $\bldu(0)$ and $\bldtheta(2n)$
can be set to any two words in $\Boolean_*^n$ that start
and end with a~$0$ and $\TCAM(\bldu(x),\bldtheta(2n)) = 0$
(e.g., we can take $\bldu(0)$ and $\bldtheta(2n)$ to be
distinct in $\Boolean^n$ that start and end with a~$0$).
Note that this is always possible when $n \ge 3$;
for $n = 1, 2$ we restrict the table to the rows
that correspond to $x, t \in [2n \rangle$.
\fi%
\end{proof}

\begin{table}[hbt]
\caption{Mappings
for $\GEset_q$ that attain $q = 2n+1$ under Scenario~(\Sdd).}
\label{tab:GEsetSdd}
\[
\ifPAGELIMIT
    \renewcommand{\arraystretch}{0.7}
\fi
\begin{array}{c|c}
x      & \bldu(x)                \\
\hline
\raisebox{0ex}[2ex]{}
0      & 0000           \ldots 0 \\
1      & 1000           \ldots 0 \\
2      & 1100           \ldots 0 \\
\vdots & \ddots                  \\
n{-}1  & 111           \ldots 10 \\
n      & 1111           \ldots 1 \\
n{+}1  & {*}111         \ldots 1 \\
n{+}2  & {*}{*}11       \ldots 1 \\
\vdots & \hspace{3ex} \ddots     \\
2n{-}1 & {*}{*}{*}   \ldots {*}1 \\
2n     & {*}{*}{*}{*} \ldots {*}
\end{array}
\quad\quad\quad
\begin{array}{c|c}
t      & \bldtheta(t)            \\
\hline
\raisebox{0ex}[2ex]{}
0      & {*}{*}{*}{*} \ldots {*} \\
1      & 1{*}{*}{*}   \ldots {*} \\
2      & 11{*}{*}     \ldots {*} \\
\vdots & \ddots                  \\
n{-}1  & 111         \ldots 1{*} \\
n      & 1111           \ldots 1 \\
n{+}1  & 0111           \ldots 1 \\
n{+}2  & 0011           \ldots 1 \\
\vdots & \hspace{3ex} \ddots     \\
2n{-}1 & 000           \ldots 01 \\
2n     & 0100           \ldots 0
\end{array}
\]
\end{table}

\ifPAGELIMIT
    A similar proof leads to our results in Table~\ref{tab:summary}
    for $\Phi = \GEset_q$
    under Scenarios~\textup{(\Sdr)} and~\textup{(\Srr)}.
\else
\begin{proposition}
\label{prop:GEsetSdr}
Under each of the scenarios~\textup{(\Sdr)} or~\textup{(\Srr)},
the family $\GEset_q$ is $n$-cell implementable, \ifandonlyif
\[
q \le 2n+1 .
\]
\end{proposition}

\begin{proof}
Necessity will follow from the proof of
Proposition~\ref{prop:GELEsetSdr} below,
and sufficiency follows from the mappings $x \mapsto \bldu(x)$
and $t \mapsto \bldtheta(t)$
defined in~(\ref{eq:simple3u})--(\ref{eq:simple3theta})
and extended to the domain $[2n+1 \rangle$ by:
\[
\bldu(2n) = {*}{*}\ldots{*}
\quad \textrm{and} \quad
\bldtheta(2n) = {\reject}{\reject}\ldots{\reject} .
\]
(For $n \ge 3$, we can also use the mappings
in Table~\ref{tab:GEsetSdd}.)
\end{proof}

\fi

\section{The family $\GEset_q \cup \LEset_q$}
\label{sec:GELEset}

\ifPAGELIMIT
    We consider here
\else
In this section, we consider
\fi
the family $\Phi = \GEset_q \cup \LEset_q$.
For the subset $\GEset_q$ of functions in this family,
we use---as in Section~\ref{sec:GEset}---the notation
$t \mapsto \bldtheta(t)$
for $\GE_q(\cdot,t) \mapsto \bldtheta(\GE_q(\cdot,t))$.
For the remaining subset $\LEset_q$,
we use $t \mapsto \bldtheta'(t)$
for $\LE_q(\cdot,t) \mapsto \bldtheta(\LE_q(\cdot,t))$
so that~(\ref{eq:TCAM}) becomes
\begin{equation}
\ifPAGELIMIT
    \ifANONYMOUS
        \label{eq:TCAM'}
    \else
    \nonumber
    \fi
\else
\label{eq:TCAM'}
\fi
\LE_q(x,t) = \TCAM(\bldu(x),\bldtheta'(t)) .
\end{equation}

\ifPAGELIMIT
\else
\begin{proposition}
\label{prop:GELEsetSdb}
Under Scenario~\textup{(\Sdb)},
the family $\GEset_q \cup \LEset_q$ is $n$-cell implementable,
\ifandonlyif
\[
q \le n.
\]
\end{proposition}

\begin{proof}
Sufficiency follows from Proposition~\ref{prop:simple1}.
As for necessity, we consider first just
the~$q$ mappings $x \mapsto \GE_q(x,t)$
in $\GEset_q \cup \LEset_q$
and refer to the containment~(\ref{eq:GEsetSdd})
in the proof of Proposition~\ref{prop:GEsetSdd}.
For Scenario~(\Sdb)
we have $y_j = 0$ for all $j \in [n \rangle$ and, so,
\[
[1:q \rangle \subseteq \bigl\{ z_j \,:\, j \in [n \rangle \bigr\} .
\]
If $z_j = \infty$ for at least one $j \in [n \rangle$ we are done.
Otherwise, by~(\ref{eq:x}),
we must have $\bldu(q-1) = {*}{*} \ldots {*}$,
yet this would imply
$\LE_q(q-1,0) = \TCAM(\bldu(q-1),\bldtheta'(0)) = 1$,
which is impossible.
\end{proof}


\begin{proposition}
\label{prop:GELEsetSbd}
Under Scenario~\textup{(\Sbd)},
the family $\GEset_q \cup \LEset_q$ is $n$-cell implementable,
\ifandonlyif
\[
q \le n + 1.
\]
\end{proposition}

\begin{proof}
Necessity follows from Proposition~\ref{prop:GEsetSbd}.
Sufficiency follows by taking the mappings $x \mapsto \bldu(x)$
and $t \mapsto \bldtheta(t)$
as in the proof of that proposition and
defining the mapping
$t \mapsto \bldtheta'(t)$ by
\[
\vartheta'_j(t) =
\left\{
\begin{array}{ll}
0   & \textrm{if $t \le j$} \\
{*} & \textrm{if $t >   j$}
\end{array}
\right.
.
\]
\end{proof}

\fi

\begin{proposition}
\label{prop:GELEsetSdd}
Under each of the scenarios~\textup{(\Sdd)} or~\textup{(\Srd)},
the family $\GEset_q \cup \LEset_q$ is $n$-cell implementable,
\ifandonlyif\
\ifPAGELIMIT
    $q \le \max \{ 2n-1, 2 \}$.
\else
\[
q \le
\left\{
\begin{array}{lcl}
2      && \textrm{if $n = 1$} \\
2n - 1 && \textrm{if $n \ge 2$}
\end{array}
\right.
.
\]
\fi
\end{proposition}

\ifPAGELIMIT
    Sufficiency (for $n \ge 2$) follows from Table~\ref{tab:GELEsetSdd}
    (where some symbols are underlined to make it easier to see
    the general pattern).
\else
\begin{proof}
Sufficiency follows from Table~\ref{tab:GELEsetSdd} for $n \ge 2$
(some symbols in the table are underlined to make it easier to see
the general pattern). As for $n = 1$, we take:
\[
\begin{array}{c|c}
x & u(x) \\
\hline
\raisebox{0ex}[2ex]{}
0 & 0 \\
1 & 1
\end{array}
\quad\quad\quad
\begin{array}{c|c|c}
t & \vartheta(t) & \vartheta'(t)         \\
\hline
\raisebox{0ex}[2ex]{}
0 & {*} & 0                              \\
1 & 1   & {*}
\end{array}
\]

We show necessity by induction on~$n$.
Suppose that Eq.~(\ref{eq:TCAM}) holds and,
for $j \in [n \rangle$, let $y_j$ and $z_j$ be
defined for the mapping $t \mapsto \bldtheta(t)$ in
as in the proof of Proposition~\ref{prop:GEsetSdd}.
Also, let $y'_j$ be the largest
$t \in [q \rangle$ such that $\vartheta'_j(t) \in \Boolean$
(define $y'_j = -\infty$ if no such~$t$ exists)
and let $z'_j$ be the largest
$t \in [y'_j \rangle$ such that
$\vartheta'_j(t) \ne \vartheta'_j(y'_j)$
(with $z'_j = -\infty$ if no such~$t$ exists).\footnote{%
In other words, if $\hat{y}_j$ and $\hat{z}_j$ are defined
for the mapping
$t \mapsto \hat{\vartheta}_j(t) = \vartheta'_j(q{-}1{-}t)$
as $y_j$ and $z_j$ were defined for $t \mapsto \vartheta_j(t)$
in the proof of Proposition~\ref{prop:GEsetSdd},
then $y'_j = q{-}1{-}\hat{y}_j$ and $z'_j = q{-}1{-}\hat{z}_j$.}
Denote by~$Y$ (respectively, $Y'$) the set of all indexes
$j \in [n \rangle$ such that $y_j$ (respectively, $y'_j$) is finite
and by $Z \; (\subseteq Y)$
(respectively, $Z' \; (\subseteq Y')$)
the set of all indexes $j \in [n \rangle$
such that $z_j$ (respectively, $z'_j$) is finite
Clearly, $|Z| \le |Y| \le n$ and $|Z'| \le |Y'| \le n$
and, by~(\ref{eq:GEsetSdd})
(when applied to
$x \mapsto \bldu(x)$ and $t \mapsto \bldtheta(t)$
on the one hand, and to
$x \mapsto \bldu(q{-}1{-}x)$ and $t \mapsto \bldtheta'(q{-}1{-}t)$
on the other hand), we get:
\begin{equation}
\label{eq:J}
q-1 \le \min \{ |Y| + |Z|, |Y'| + |Z'| \} .
\end{equation}

If $Y = Z = [n \rangle$
then necessarily $\bldu(q-1) = {*}{*} \ldots {*}$ (by~(\ref{eq:x})),
yet this would imply by~(\ref{eq:TCAM})
that $\LE_q(q-1,0) = 1$, which is impossible.
Hence, $|Y| + |Z|\le 2n-1$ and, similarly, $|Y'| + |Z'| \le 2n-1$.
We conclude that the right-hand side of~(\ref{eq:J})
is at most $2n-1$, thereby establishing the necessity condition
for the induction base $n = 1$,
and also when $\min \{ |Y| + |Z|, |Y'| + |Z'| \} \le 2n-2$.

It remains to consider the case where
$Y = Y' = [n \rangle$ and $|Z| = |Z'| = n-1$.
By~(\ref{eq:x}), the entries of $\bldu(q-1)$ are
then all~$*$ except one,
and the same applies to $\bldu(0)$. Without loss of generality
we can assume that $Z = [n-1 \rangle$, thereby implying that
\[
\bldu(q-1) = {*}{*} \ldots {*} \reject
\quad \textrm{or} \quad
\bldu(q-1) = {*}{*} \ldots {*} 1 .
\]
In the first case we must have $\vartheta_{n-1}(t) = {*}$
for all $t \in [q \rangle$, thus effectively reducing
the value of~$n$ by~$1$ when writing~(\ref{eq:TCAM})
for the functions in $\GEset_q$ within $\GEset_q \cup \LEset_q$;
by Proposition~\ref{prop:GEsetSdd} we then get
$q \le 2(n-1) + 1 = 2n-1$.
In the remaining part of the proof we assume
that $\bldu(q-1) = {*}{*} \ldots {*} 1$;
this, in turn, forces having
\begin{equation}
\label{eq:n-1,q-1}
\vartheta'_{n-1}(t) = 0 \;\; \textrm{for all $t \in [q-1 \rangle$} .
\end{equation}
We distinguish between
the two possible values for $\vartheta'_{n-1}(q-1)$.

\emph{Case 1: $\vartheta'_{n-1}(q-1) = 1$.}
Here we get from~(\ref{eq:n-1,q-1}) that
$u_{n-1}(x) = {*}$ for all $x \in [q-1 \rangle$,
which means that Eq.~(\ref{eq:TCAM})
holds for all $f \in \GEset_q \cup \LEset_q$
with~$q$ and~$n$ in~(\ref{eq:TCAM}) replaced by
$q-1$ and $n-1$, respectively.
Hence, by the induction hypothesis,
\[
q-1 \le \max \{ 2(n-1) - 1, 2 \} ,
\]
namely, $q \le \max \{ 2n-2, 3 \} \le 2n-1$ when $n \ge 2$.

\emph{Case 2: $\vartheta'_{n-1}(q-1) = {*}$.}
Here we get from~(\ref{eq:n-1,q-1}) that
$Z' = Z = [n-1 \rangle$,
which means that the unique \mbox{non-$*$}
entry in $\bldu(0)$ must be its last (i.e., co-located with
the unique \mbox{non-$*$} entry in $\bldu(q-1)$).
Recalling that $\vartheta'_{n-1}(0) = 0$
(by~(\ref{eq:n-1,q-1})), we must then have
$\bldu(0) = {*}{*} \ldots {*} 0$, thereby implying that
\[
\vartheta_{n-1}(t) = 1
\;\; \textrm{for all $t \in [1:q \rangle$} .
\]
Combining this with~(\ref{eq:n-1,q-1})
we must also have $u_{n-1}(x) = {*}$ for all $x \in [1:q-1 \rangle$.
We conclude that Eq.~(\ref{eq:TCAM}) holds
for all $f \in \GEset_q \cup \LEset_q$
with~$n$ in that equation replaced by $n-1$
and with~$x$ and~$t$ restricted to $[1:q-1 \rangle$.
Replacing~$x$ and~$t$ in~(\ref{eq:TCAM})
by $x+1$ and $t+1$, respectively, and noting that
$\GE_{q-2}(x,t) = \GE_q(x+1,t+1)$
and
$\LE_{q-2}(x,t) = \LE_q(x+1,t+1)$
for $x, t \in [q-2 \rangle$, we get from the induction hypothesis that
\[
q-2 \le \max \{ 2(n-1)-1, 2 \} ,
\]
namely, $q \le \max \{ 2n-1, 4 \}$.
This completes the proof for Case~2 when $n > 2$.
As for $n = 2$, if $q = 4$ were attainable, then, from
what we have just shown, we would have $u_1(1) = u_1(2) = {*}$,
which would imply that
$u_0(1) = \vartheta'_0(1) \in \Boolean$
and $u_0(2) = \vartheta_0(2) \in \Boolean$
and $u_0(1) \ne u_0(2)$.
However, this would
mean that $\LE_4(x,0) = \TCAM(u_0(x),\vartheta'_0(0)) = 1$
for at least one $x \in \{ 1, 2 \}$
(regardless of the value of $\vartheta'_0(0)$),
which is a contradiction.
\end{proof}
\fi

\begin{table}[hbt]
\caption{Mappings for $\GEset_q \cup \LEset_q$
that attain $q = 2n-1$ under Scenario~(\Sdd).%
\ifPAGELIMIT
    \vspace{-2ex}
\fi}
\label{tab:GELEsetSdd}
\[
\ifPAGELIMIT
    \renewcommand{\arraystretch}{0.7}
\fi
\begin{array}{c|c}
x      & \bldu(x) \\
\hline
\raisebox{0ex}[2ex]{}
0      & \Emph{00}{*}{*}{*}{*}     \ldots {*} \\
\hline
\raisebox{0ex}[2ex]{}
1      & \Emph{1}0{*}{*}{*}{*}     \ldots {*} \\
2      & \Emph{1}{*}0{*}{*}{*}     \ldots {*} \\
3      & \Emph{1}{*}{*}0{*}{*}     \ldots {*} \\
\vdots & \hspace{0.4ex}\Emph{\vdots}
         \hfill \hspace{1ex} \ddots \hfill    \\
n{-}2  & \Emph{1}{*}{*}{*}     \ldots {*}0{*} \\
n{-}1  & \Emph{1}{*}{*}{*}{*}     \ldots {*}0 \\
\cline{1-2}
\raisebox{0ex}[2ex]{}
n      & \Emph{11}{*}{*}{*}{*}     \ldots {*} \\
\cline{1-2}
\raisebox{0ex}[2ex]{}
n{+}1  & \Emph{{*}1}1{*}{*}{*}     \ldots {*} \\
n{+}2  & \Emph{{*}1}{*}1{*}{*}     \ldots {*} \\
n{+}3  & \Emph{{*}1}{*}{*}1{*}     \ldots {*} \\
\vdots & \hspace{1.4ex}\Emph{\vdots} \hfill\hspace{2ex} \ddots \hfill \\
2n{-}3 & \Emph{{*}1}{*}{*}     \ldots {*}1{*} \\
2n{-}2 & \Emph{{*}1}{*}{*}{*}     \ldots {*}1
\end{array}
\quad\quad\quad
\begin{array}{c|c|c}
t      & \bldtheta(t) & \bldtheta'(t)         \\
\hline
\raisebox{0ex}[2ex]{}
0      & \Emph{{*}}{*}{*}{*}{*}{*} \ldots {*}
       & \Emph{00}0000               \ldots 0 \\
\hline
\raisebox{0ex}[2ex]{}
1      & \Emph{1}{*}{*}{*}{*}{*}   \ldots {*}
       & \Emph{{*}0}1111             \ldots 1 \\
2      & \Emph{1}1{*}{*}{*}{*}     \ldots {*}
       & \Emph{{*}0}0111             \ldots 1 \\
3      & \Emph{1}11{*}{*}{*}       \ldots {*}
       & \Emph{{*}0}0011             \ldots 1 \\
\vdots & \ddots & \hspace{5ex} \ddots         \\
n{-}2  & \Emph{1}111           \ldots 1{*}{*}
       & \Emph{{*}0}000             \ldots 01 \\
n{-}1  & \Emph{1}1111             \ldots 1{*}
       & \Emph{{*}0}0000             \ldots 0 \\
\cline{1-1} \cline{3-3}
\raisebox{0ex}[2ex]{}
n      & \Emph{1}11111               \ldots 1
       & \Emph{{*}}{*}0000           \ldots 0 \\
\cline{1-2}
\raisebox{0ex}[2ex]{}
n{+}1  & \Emph{01}1111               \ldots 1
       & \Emph{{*}}{*}{*}000         \ldots 0 \\
n{+}2  & \Emph{01}0111               \ldots 1
       & \Emph{{*}}{*}{*}{*}00       \ldots 0 \\
n{+}3  & \Emph{01}0011               \ldots 1
       & \Emph{{*}}{*}{*}{*}{*}0     \ldots 0 \\
\vdots & \hspace{4ex} \ddots & \hspace{6ex} \ddots         \\
2n{-}3 & \Emph{01}00               \ldots 011
       & \Emph{{*}}{*}{*}{*}{*}   \ldots {*}0 \\
2n{-}2 & \Emph{01}000 \ldots 01
       & \Emph{{*}}{*}{*}{*}{*}{*} \ldots {*}
\end{array}
\]
\end{table}

\begin{proposition}
\label{prop:GELEsetSdr}
Under each of the scenarios~\textup{(\Sdr)} or~\textup{(\Srr)},
the family $\GEset_q \cup \LEset_q$ is $n$-cell implementable,
\ifandonlyif\
\ifPAGELIMIT
    $q \le 2n$.
\else
\[
q \le 2n .
\]
\fi
\end{proposition}

\ifPAGELIMIT
    Sufficiency follows from Proposition~\ref{prop:simple3},
    and necessity follows from arguments that are similar to
    those that were used in the proof
    Proposition~\ref{prop:GEsetSdd}.
\else
\begin{proof}
Sufficiency follows from Proposition~\ref{prop:simple3}.

Necessity follows from arguments that are similar to
(and in fact are even simpler than)
those that we used in the proofs of
Propositions~\ref{prop:GEsetSdd} and~\ref{prop:GELEsetSdd}.
Specifically, for each $j \in [n \rangle$, let $y_j$ denote the smallest
$t \in [q \rangle$ such that
$\vartheta_j(t) \in \Boolean \cup \{ \reject \}$
(with $y_j = \infty$ if no such~$t$ exists).
By the very same reasoning that was given
in the proof of Proposition~\ref{prop:GEsetSdd},
we can assume hereafter without loss of generality that
$\vartheta_j(y_j) \in \{ 1, \reject \}$
and that
$\vartheta_j(t) \in \Boolean \cup \{ \reject \}$ for every $t > y_j$.

Next, for each $j \in [n \rangle$, let $z_j$ be the smallest
$t \in [y_j:q \rangle$ such that
$\vartheta_j(t) \in \{ 0, \reject \}$
(with $z_j = \infty$ if no such~$t$ exists).
By~(\ref{eq:x}) we then must have
$u_j(x) = {*}$ for every $x \ge z_j$
and, so, without loss of generality we can assume that
$\vartheta_j(t) = \reject$ for every $t \ge z_j$.
In summary, the mapping $t \mapsto \bldtheta(t)$
takes the following form:
\begin{equation}
\label{eq:Qyjzj}
\vartheta_j(t) =
\left\{
\begin{array}{lcl}
{*}     && \textrm{if $t < y_j$} \\
1       && \textrm{if $y_j \le t < z_j$} \\
\reject && \textrm{if $t \ge z_j$}
\end{array}
\right.
.
\end{equation}

As our next step, we show that~(\ref{eq:GEsetSdd}) holds.
Indeed, if there were $y \in [1:q \rangle$ that did not belong
to the right-hand side of~(\ref{eq:GEsetSdd})
then, from~(\ref{eq:Qyjzj}), we would have
$\bldtheta(y) = \bldtheta(y-1)$.
Yet, by~(\ref{eq:TCAM}), this would mean
that $\GE_q(y-1,y) = \GE_q(y,y) = 1$, which is a contradiction.
  From~(\ref{eq:GEsetSdd}) we now get the necessary condition in
Proposition~\ref{prop:GEsetSdr} (for the family $\Phi = \GEset_q$).
The necessary condition for
$\Phi = \GEset_q \cup \LEset_q$
follows by showing that $z_j = \infty$
for at least one $j \in [n \rangle$ (and, therefore,
the right-hand side of~(\ref{eq:GEsetSdd}) contains less
than $2n$ elements).
Otherwise, we would have
$\bldtheta(q-1) = \reject \reject \ldots \reject$
and, consequently, $\bldu(q-1) = {*}{*} \ldots {*}$,
which, with~(\ref{eq:TCAM'}), would
yield $\LE_q(q-1,0) = 1$, thereby reaching a contradiction.
\end{proof}
\fi

\section{The whole set $\Fset_q$}
\label{sec:Fset}

\ifPAGELIMIT
    \begin{proposition}
    \label{prop:FsetSdd}
    Under each of the scenarios~\textup{(\Sdb)}, \textup{(\Sbd)},
    \textup{(\Sdd)}, or~\textup{(\Srd)},
    the set $\Fset_q$ is $n$-cell implementable, \ifandonlyif\
    $q \le n$.
    \end{proposition}

    Sufficiency follows from Proposition~\ref{prop:simple1}.

    Proposition~\ref{prop:FsetSdd} can also be stated
    in terms of the VC~dimension of Boolean monomials
    under extended substitutions rules~\cite[\S 7.3]{AB},\cite{NS}.
    Specifically, in addition to the elements of~$\Boolean$,
    we allow to substitute~$*$ (respectively, $\reject$) into each
    variable, in which case both the variable and its complement
    are defined to be~$1$ (respectively, $0$).
    Scenario~(\Sbd) corresponds to the case
    where the evaluation points are restricted to $\Boolean^n$
    and, for this scenario, it was shown in~\cite{NS} that
    the VC~dimension of the set of all $n$-variate monomials equals~$n$.
    Proposition~\ref{prop:FsetSdd} implies that the VC~dimension
    does not increase even if we extend
    the set of evaluation points to $\Boolean_\reject^n$.
\else
In this section, we treat the case where $\Phi = \Fset_q$
(the whole set of functions $[q \rangle \rightarrow \Boolean$).
Differently from previous sections, we start with
Scenarios~(\Srd) and~(\Srr), as they are rather straightforward.

\begin{proposition}
\label{prop:FsetSdr}
Under each of the scenarios~\textup{(\Sdr)} or~\textup{(\Srr)},
the set $\Fset_q$ is $n$-cell implementable, \ifandonlyif\
\[
q \le 2n .
\]
\end{proposition}

\begin{proof}
Sufficiency follows from Proposition~\ref{prop:simple3}
and necessity follows either from Proposition~\ref{prop:NsetSdd} or
by a counting argument:
$\bldtheta : \Fset_q \rightarrow \Boolean_\reject^n$ is injective
and, so, $2^q \le 2^{2n}$.
\end{proof}

The next proposition covers all the remaining scenarios
in Table~\ref{tab:summary}.

\begin{proposition}
\label{prop:FsetSdd}
Under each of the scenarios~\textup{(\Sdb)}, \textup{(\Sbd)},
\textup{(\Sdd)}, or~\textup{(\Srd)},
the set $\Fset_q$ is $n$-cell implementable, \ifandonlyif
\[
q \le n .
\]
\end{proposition}

Sufficiency follows from Proposition~\ref{prop:simple1}.
As for necessity, for Scenarios~\textup{(\Sdb)} and~\textup{(\Sbd)}
it is implied by Proposition~\ref{prop:NsetSbd}
(for Scenario~\textup{(\Sdb)} we can also use
Proposition~\ref{prop:GELEsetSdb} or just a counting argument).
For Scenarios~(\Sdd) and~(\Srd), however, some more effort is needed;
notice that a counting argument only leads to
$2^q \le 3^n$, namely, to the weaker inequality
$q \le n \cdot \log_2 3$.

The proof of Proposition~\ref{prop:FsetSdd}
will use the following notation and lemma.
Recalling the partial ordering~$\preceq$ of
Section~\ref{sec:Eset},
for elements $u, v \in \Boolean_\reject$,
we denote by $\mu(u,v)$ the largest element
$s \in \Boolean_\reject$ such that both
$s \preceq u$ and $s \preceq v$, where ``largest'' is
with respect to~$\preceq$. Thus, for every $u \in \Boolean_\reject$,
\[
\mu(u,\reject) = \reject , \quad
\mu(u,{*}) = \mu(u,u) = u , \quad \textrm{and} \quad
\mu(0,1) = \reject .
\]
The next lemma is easily verified.

\begin{lemma}
\label{lem:FsetSdd}
For every $u, v, \vartheta \in \Boolean_\reject$,
\[
\TCAM(\mu(u,v),\vartheta) =
\TCAM(u,\vartheta) \wedge \TCAM(v,\vartheta) .
\]
\end{lemma}

\begin{proof}[Proof of Proposition~\protect\ref{prop:FsetSdd}
(necessity)]
We prove necessity under (the loosest) Scenario~(\Srd)
by induction on~$q$, with the induction
base ($q = 2$) following from a simple counting argument:
there are four distinct functions
$f : [2 \rangle \rightarrow \Boolean$ yet only
three choices for (the scalar) $\vartheta(f)$ in this case,
hence we must have $n \ge 2$.\footnote{%
\label{footnote:Fsetq=2}
Note that $n = 1$ would suffice for $q = 2$ if we allowed
$\vartheta(f)$ to take the value~$\reject$ as well:
this corresponds to Scenario~(\Sbr)
and it differs from Scenario~(\Sbd) only when $n = 1$ and $q = 2$.}

Turning to the induction step,
suppose to the contrary that $\Fset_q$ is $n$-cell implementable
for $q > n$ and
let $\bldu : [q \rangle \rightarrow \Boolean_\reject^n$
and $\bldtheta : \Fset_q \rightarrow \Boolean_*^n$
be mappings such that~(\ref{eq:TCAM}) holds for all
the functions in $\Fset_q$;
obviously, both these mappings are injective.

Consider the images of the functions in $\Nset_q$
under~$\bldtheta$. As argued in the proof of
Proposition~\ref{prop:NsetSbd},
we can assume that for each $t \in [q \rangle$
there is a unique entry in $\bldtheta(\N_q(\cdot,t))$
which is \mbox{non-$*$}. From $q > n$
we get, by the pigeonhole principle,
that there exist two distinct elements $t_0, t_1 \in [q \rangle$
for which the position~$j$ of such an entry is the same, say $j = n-1$.
Without loss of generality we can assume further
that $t_0 = q-2$, $t_1 = q-1$,
\[
\vartheta_{n-1}(\N_q(\cdot,q-2)) = 0 ,
\quad \textrm{and} \quad
\vartheta_{n-1}(\N_q(\cdot,q-1)) = 1 .
\]
This, in turn, implies that for every $x \in [q \rangle$:
\begin{equation}
\label{eq:un-1}
u_{n-1}(x) =
\left\{
\begin{array}{lcl}
{*} && \textrm{if $x \in [q-2 \rangle$} \\
1   && \textrm{if $x = q-2$} \\
0   && \textrm{if $x = q-1$}
\end{array}
\right.
.
\end{equation}

Next, we define the mappings
\[
\bldu' : [q-1 \rangle \rightarrow \Boolean_\reject^{n-1}
\quad \textrm{and} \quad
\bldtheta' : \Fset_{q-1} \rightarrow \Boolean_*^{n-1}
\]
as follows:
for every $x \in [q-1 \rangle$ and $j \in [n-1 \rangle$,
\begin{equation}
\label{eq:u'j}
u'_j(x) =
\left\{
\begin{array}{cl}
u_j(x)                  & \textrm{if $x \in [q-2 \rangle$} \\
\mu(u_j(q-2), u_j(q-1)) & \textrm{if $x = q-2$}
\end{array}
\right.
,
\end{equation}
and for every $f \in \Fset_{q-1}$ and $j \in [n-1 \rangle$,
\begin{equation}
\label{eq:theta'j}
\vartheta'_j(f) = \vartheta_j(\tilde{f}) ,
\end{equation}
where~$\tilde{f}$ is the extension of~$f$
to the domain $[q \rangle$ with
\[
\tilde{f}(q-1) = f(q-2) .
\]
In particular, it follows from
(\ref{eq:u'j}), (\ref{eq:theta'j}), and Lemma~\ref{lem:FsetSdd}
that for every $j \in [n-1 \rangle$:
\begin{eqnarray}
\lefteqn{
\TCAM(u'_j(q{-}2),\vartheta'_j(f))
} \makebox[2ex]{} \nonumber \\
\label{eq:x=q-2}
& = &
\!\!\!
\TCAM(u_j(q{-}2),\vartheta_j(\tilde{f}))
\wedge \TCAM(u_j(q{-}1)),\vartheta_j(\tilde{f})) .
\end{eqnarray}

We show that
\[
f'(x) = \TCAM(\bldu'(x),\bldtheta'(f))
\]
for every $x \in [q-1 \rangle$ and $f \in \Fset_{q-1}$
which, in turn, will imply that $\Fset_{q-1}$
is $(n-1)$-cell implementable, thereby contradicting
the induction hypothesis. We distinguish between three cases.

\emph{Case 1: $x \in [q-2 \rangle$.}
We recall from~(\ref{eq:un-1}) that $u_{n-1}(x) = {*}$ and, so,
\[
f(x) = \tilde{f}(x)
\stackrel{\textrm{(\ref{eq:TCAM})}}{=}
\TCAM(\bldu(x),\bldtheta(\tilde{f}))
\stackrel{\textrm{(\ref{eq:u'j})+(\ref{eq:theta'j})}}{=}
\TCAM(\bldu'(x),\bldtheta'(f)) .
\]

\emph{Case 2: $x = q-2$ and $f'(q-2) = 1$.}
Here $\tilde{f}(q-2) = \tilde{f}(q-1) = 1$ and,
therefore, from~(\ref{eq:TCAM}),
for every $j \in [n \rangle$,
\[
\TCAM(u_j(q-2),\vartheta_j(\tilde{f}))
= \TCAM(u_j(q-1),\vartheta_j(\tilde{f})) = 1 .
\]
Hence, by~(\ref{eq:x=q-2}), for every $j \in [n-1 \rangle$,
\[
\TCAM(u'_j(q-2),\vartheta'_j(f)) = 1 ,
\]
namely, $\TCAM(\bldu'(q-2),\bldtheta'(f)) = 1 = f(q-2)$.

\emph{Case 3: $x = q-2$ and $f(q-2) = 0$.}
Here $\tilde{f}(q-2) = \tilde{f}(q-1) = 0$;
yet, from~(\ref{eq:un-1})
we have $\TCAM(u_{n-1}(y),\vartheta_{n-1}(\tilde{f})) = 1$ for at least
one $y \in \{ q-2, q-1 \}$. Hence, from
$\TCAM(\bldu(y),\bldtheta(\tilde{f})) = \tilde{f}(y) = 0$
there must be at least one index $j \in [n-1 \rangle$ for which
\[
\TCAM(u_j(y),\vartheta_j(\tilde{f})) = 0 ,
\]
and for that index we have, by~(\ref{eq:x=q-2}),
\[
\TCAM(u'_j(q-2),\vartheta'_j(f)) = 0 ,
\]
namely, $\TCAM(\bldu'(q-2),\bldtheta'(f)) = 0 = f(q-2)$.
\end{proof}

\subsection{Connection to the VC~dimension of Boolean monomials}
\label{sec:VCdimension}

Proposition~\ref{prop:FsetSdd} can be stated also
in terms of the VC~dimension of the following collection,
$\Collection_n$, of $3^n$ subsets of $\Boolean_\reject^n$:
\[
\Collection_n =
\bigl\{
\Sphere = \Sphere(\bldtheta) \,:\, \bldtheta \in \Boolean_*^n \bigr\} ,
\]
where, for each $\bldtheta \in \Boolean_*^n$,
\[
\Sphere(\bldtheta) =
\bigl\{ \bldv \in \Boolean_\reject^n \,:\, \TCAM(\bldv,\bldtheta) = 1
\bigr\} .
\]
We demonstrate this next.

Let $\bldu : [q \rangle \rightarrow \Boolean_\reject^n$
and $\bldtheta : \Fset_q \rightarrow \Boolean_*^n$
be injective mappings. The following two conditions are
equivalent for any $f \in \Fset_q$.
\begin{itemize}
\item
Eq.~(\ref{eq:TCAM}) holds for~$f$.
\item
The images of $x \mapsto \bldu(x)$ form a subset
$\U \subseteq \Boolean_\reject^n$ of size~$q$ such that
\begin{equation}
\label{eq:shatter}
\U \cap \Sphere(\bldtheta(f)) =
\{ \bldu(x) \,:\,
x \in [q \rangle \; \textrm{such that} \; f(x) = 1 \} .
\end{equation}
\end{itemize}
In particular, if Eq.~(\ref{eq:TCAM}) holds for every $f \in \Fset_q$,
then~(\ref{eq:shatter})
implies that~$\U$ is \emph{shattered} by $\Collection_n$:
each of the $2^q$ subsets of~$\U$ can be expressed
as an intersection $\U \cap \Sphere$, for some
$\Sphere \in \Collection_n$.
Conversely, if~$\U$ is a subset of size~$q$ of $\Boolean_\reject^n$
that is shattered by $\Collection_n$,
we can fix some arbitrary bijection $\bldu : [q \rangle \rightarrow \U$
and define a mapping $\bldtheta : \Fset_q \rightarrow \Boolean_*^n$
so that~(\ref{eq:shatter}) holds for every $f \in \Fset_q$;
specifically, we select $\bldtheta(f)$ to be such that
$\Sphere(\bldtheta(f))$ is an element
$\Sphere \in \Collection_n$ for which $\U \cap \Sphere$
equals the subset of~$\U$ given by
the right-hand side of~(\ref{eq:shatter}).
The largest size~$q$ of~$\U$ for which this holds is
the VC~dimension of~$\Collection_n$~\cite[\S 7.3]{AB}.
Thus, Proposition~\ref{prop:FsetSdd}
states that the VC~dimension of~$\Collection_n$
is~$n$.

Equivalently, we can state Proposition~\ref{prop:FsetSdd}
in terms of the VC~dimension of
the collection of $n$-variate Boolean monomials where
the evaluation points are taken from $\Boolean_\reject^n$
(rather than just from $\Boolean^n$).
Specifically, given a vector of $n$ Boolean indeterminates,
$\bldxi = (\xi_j)_{j \in [n \rangle}$,
we associate with every word
$\bldtheta = (\vartheta_j)_{j \in [n \rangle} \in \Boolean_*^n$
the $n$-variate Boolean monomial
\[
M_\bldtheta(\bldxi)
=
\Bigl(
\bigwedge_{j \,:\, \vartheta_j = 0} \xi_j
\Bigr)
\wedge
\Bigl(
\bigwedge_{j \,:\, \vartheta_j = 1} \overline{\xi}_j
\Bigr)
,
\]
where $\overline{\xi}_j$ stands for the complement of $\xi_j$.
Substituting an element of~$\Boolean$ into a variable $\xi_j$
carries its ordinary meaning
(with $\overline{0} = 1$ and $\overline{1} = 0$),
whereas when substituting~$*$ (respectively, $\reject$),
both $\xi_j$ and $\overline{\xi}_j$ are defined to be~$1$
(respectively, $0$).
Under these rules, $\Sphere(\bldtheta)$ is the set of
all words in $\Boolean_\reject^n$ at which
$M_\bldtheta(\bldxi)$ evaluates to~$1$.
Scenario~(\Sbd) corresponds to the case
where the evaluation points are restricted to $\Boolean^n$
and, for this scenario, it was shown in~\cite{NS} that
the VC~dimension of the set of all $n$-variate monomials equals~$n$.
Proposition~\ref{prop:FsetSdd} implies that the VC~dimension
does not increase even if we extend
the set of evaluation points to $\Boolean_\reject^n$.\footnote{%
A setting where some variables are allowed to
be ``unspecified'' has been studied in
the literature (see, for example,~\cite{BKKS}),
yet the meaning of~$*$ therein is different.}

\subsection{Subsets that are as hard to implement as $\Fset_q$}
\label{sec:hardsubsets}
\fi

Recall that under Scenarios~(\Sdb), (\Sbd), (\Sdr), and~(\Srr),
there are small subsets of $\Fset_q$
which are $n$-cell implementable, (if and) only if $\Fset_q$ is;
e.g., $\Nset_q$ is such as subset of size $q = \log_2 |\Fset_q|$.
In contrast, it turns out that
under Scenarios~(\Sdd) and~(\Srd), the condition $q \le n$
becomes necessary only for fairly large subsets of $\Fset_q$.
\ifPAGELIMIT
    Specifically, we can show
\else
The next two propositions (which we prove
in Appendix~\ref{sec:Fset'}) imply
\fi
that, with very few exceptions,
a deletion of just a single function from $\Fset_q$
results in a subset which is $(q-1)$-cell implementable.
\ifPAGELIMIT
    The smallest \emph{unique} subset of $\Fset_q$ that
    is \emph{not} $(q-1)$-cell implementable under Scenario~(\Sdd) is
    \[
    \Fset_q \setminus \bigl( \Eset_q \cup \{ \allone_q \} \bigr) ,
    \]
    where $x \mapsto \allone_q(x)$ is
    the constant-$1$ function over $[q \rangle$;
    under Scenario~(\Srd), the respective smallest subset
    is $\Fset_q \setminus \{ \allone_q \}$.

    \begin{proposition}
    \label{prop:FsetSdr}
    Under each of the scenarios~\textup{(\Sdr)} or~\textup{(\Srr)},
    the set $\Fset_q$ is $n$-cell implementable, \ifandonlyif\
    $q \le 2n$.
    \end{proposition}

    \begin{proof}
    Sufficiency follows from Proposition~\ref{prop:simple3}
    and necessity follows by a counting argument:
    $\bldtheta : \Fset_q \rightarrow \Boolean_\reject^n$ is injective
    and, so, $2^q \le 2^{2n}$.
    \end{proof}
\else

Henceforth, $x \mapsto \allone_q(x)$ stands for
the tautology function over $[q \rangle$ (which evaluates to~$1$
on all the elements of $[q \rangle$).

\begin{proposition}
\label{prop:Fset'Sdd}
Let $\Phi = \Fset_q \setminus \{ g \}$
where~$g$ is any function in $\Fset_q$
that is not in $\Eset_q \cup \{ \allone_q \}$.
Under Scenario~\textup{(\Sdd)}, the subset~$\Phi$
is $n$-cell implementable, \ifandonlyif
\[
q \le n+1 .
\]
\end{proposition}

\begin{proposition}
\label{prop:Fset'Srd}
Let $\Phi = \Fset_q \setminus \{ g \}$
where~$g$ is any function in $\Fset_q \setminus \{ \allone_q \}$.
Under Scenario~\textup{(\Srd)}, the subset~$\Phi$
is $n$-cell implementable, \ifandonlyif
\[
q \le n+1 .
\]
\end{proposition}

It readily follows from Proposition~\ref{prop:Fset'Sdd}
that under Scenario~(\Sdd),
any subset of~$\Fset_q$ of size smaller than $2^q - q$
is $(q-1)$-cell implementable, except (possibly) for the subset
\begin{equation}
\label{eq:Fsetstar}
\Fset_q^\star
= \Fset_q \setminus \bigl( \Eset_q \cup \{ \allone_q \} \bigr) .
\end{equation}
And by the following result (which we also prove
in Appendix~\ref{sec:Fset'}), this subset is indeed an exception.

\begin{proposition}
\label{prop:Fsetstar}
For $q \ge 3$, the subset $\Fset_q^\star$ in~(\ref{eq:Fsetstar})
is $n$-cell implementable under Scenario~(\Sdd),
(if and) only if $\Fset_q$ is.
\end{proposition}

Similar claims can be stated for Scenario~(\Srd),
yet with the subset in~(\ref{eq:Fsetstar}) replaced by
\[
\Fset_q \setminus \{ \allone_q \} .
\]
Specifically, all other proper subsets of~$\Fset_q$
are $(q-1)$-cell implementable.
\fi

\ifPAGELIMIT
\else
\section{Discussion}
\label{sec:discussion}

We can see from Table~\ref{tab:summary} that under all scenarios
and for all families of functions therein except $\Eset_q$,
the number $n$ of cells must grow linearly with the alphabet size~$q$.
While this is expected for the whole set $\Fset_q$ (simply because of
a counting argument), this also holds for the families
$\Nset_q$, $\GEset_q$, $\LEset_q$
(and $\GEset_q \cup \LEset_q$), which are only of size $O(q)$.
This means that for the latter families, CAM implementations are
necessarily exponential in the (bit) representation
of the parameter $t \in [q \rangle$ that identifies
each function in the family, as well as in the representation of
the input $x \in [q \rangle$ to each function.
Thus, some hardware modification is inevitable if more efficient
implementations are sought.

And this is indeed possible. For example, suppose that $q = 2n$
and we wish to implement each function $x \mapsto \N_{q^2}(x,t)$
in $\Nset_{q^2}$ under Scenario~(\Sdd), (\Srd), (\Sdr), or~(\Srr).
We can express each $x, t \in [q^2 \rangle$ as
\[
x = x_1 q + x_0
\quad \textrm{and} \quad
t = t_1 q + t_0 ,
\]
where $x_0, x_1, t_0, t_1 \in [q \rangle$, and observe that
\begin{equation}
\label{eq:Nsquaring}
\N_{q^2}(x,t) = \N_q(x_0,t_0) \vee \N_q(x_1,t_1) .
\end{equation}
Thus, while we have squared the alphabet size,
the right-hand side of~(\ref{eq:Nsquaring}) requires only $q = 2n$ cells
(and not $q^2/2 = 2n^2$ cells), yet such cells have to be arranged in
two $n$-cell blocks that are connected \emph{serially}
(to compute the disjunction~$\vee$), as opposed
to the ordinary \emph{parallel} connection among
the cells along a row in the CAM.

\fi

\ifIEEE
   \appendices
\else
   \section*{$\,$\hfill Appendices\hfill$\,$}
   \appendix
\fi

\ifPAGELIMIT
\else
\section{Special cases in Proposition~\protect\ref{prop:GEsetSdd}}
\label{sec:specialcases}

We verify here the necessary condition $q \le 2$ for $n = 1$,
and $q \le 4$ for $n = 2$.

Starting with $n = 1$, suppose to the contrary
that Eq.~(\ref{eq:TCAM}) can hold for $q = 2n + 1 = 3$.
By~(\ref{eq:GEsetSdd}) we then have $y_0 = 1$ and $z_0 = 2$,
namely, $\vartheta(0) = {*}$,
$\vartheta(1) = 1$, and $\vartheta(2) = 0$.
Yet this means that $\GE_3(0,t) = \TCAM(u(0), \vartheta(t)) = 1$
for at least one $t \in \{ 1, 2 \}$
(regardless of the value of $u(0)$),
which is absurd.

Turning to $n = 2$, suppose to the contrary
that Eq.~(\ref{eq:TCAM}) holds for $q = 2n + 1 = 5$.
Again, by~(\ref{eq:GEsetSdd})
we have $\{ 1, 2, 3, 4 \} = \{ y_0, z_0, y_1, z_1 \}$,
where $y_0 < z_0$ and $y_1 < z_1$.
Without loss of generality we assume that $z_0 > z_1$,
in which case $z_0 = 4$ and, so
$\vartheta_0(3) = 1$ and $\vartheta_0(4) = 0$
(see~(\ref{eq:t})).
Now, if $z_1 < 3$ then, by~(\ref{eq:x}),
we would have $u_1(2) = {*}$; yet then
$\GE_5(2,t) = \TCAM(u_0(2),\vartheta_0(t)) = 1$
for at least one $t \in \{ 3, 4 \}$, which is impossible.
Hence, $z_1 = 3$, and from~(\ref{eq:t})
we get $\vartheta_1(2) = 1$ and $\vartheta_1(3) = 0$;
moreover, $\{ y_0, y_1 \} = \{ 1, 2 \}$
and, so, $\vartheta_0(2) = 1$.
In summary, we have shown that
$\bldtheta(2) = 11$, $\bldtheta(3) = 10$,
and $\vartheta_0(4) = 0$.
Now, since $\GE_5(x,t) = \TCAM(\bldu(x),\bldtheta(t)) = 0$
for $x \in \{ 0, 1 \}$ and $t \in \{ 2, 3 \}$,
we must have $u_0(0) = u_0(1) = 0 \; (= \vartheta_0(4))$,
which implies that $u_1(0) \ne u_1(1)$. But this means
that $\GE_5(x,4) = \TCAM(u_1(x),\vartheta_1(4)) = 1$
for at least one $x \in \{ 0, 1 \}$, which is
a contradiction.

\section{More on the whole set $\Fset_q$}
\label{sec:Fset'}

We provide in this appendix
the proofs of Propositions~\ref{prop:Fset'Sdd},
\ref{prop:Fset'Srd}, and~\ref{prop:Fsetstar}.

We will sometimes find it convenient to represent a function
$f \in \Fset_q$ by its truth table
$\bldf = (f(x))_{x \in [q \rangle} \in \Boolean^q$;
for example, $11 \ldots 1$
represents the tautology function $x \mapsto \allone_q(x)$
over $[q \rangle$ and $00 \ldots 0$
represents the all-zero function, which we denote by
$x \mapsto \allzero_q(x)$.

The next lemma provides the necessity part
in Proposition~\ref{prop:Fset'Srd}
(and, therefore, also in Proposition~\ref{prop:Fset'Sdd}).

\begin{lemma}
\label{lem:Fset'Sdd}
Let $\Phi = \Fset_q \setminus \{ g \}$
where~$g$ is any function in $\Fset_q \setminus \{ \allone_q \}$.
Under Scenario~\textup{(\Srd)},
the subset~$\Phi$ is $n$-cell implementable, only if
\[
q \le n+1 .
\]
\end{lemma}

\begin{proof}
Suppose that~(\ref{eq:TCAM}) holds for the functions in~$\Phi$
with the mappings
$\bldu' : [q \rangle \rightarrow \Boolean_\reject^n$
and $\bldtheta' : \Phi \rightarrow \Boolean_*^n$.
Since $g \ne \allone_q$, there exists
a function $\overline{g} \in \Phi$ which is identical to~$g$
on the whole domain $[q \rangle$ except
for one element $y \in [q \rangle$
at which $g(y) = 0$ yet $\overline{g}(y) = 1$.

Define the mappings
$\bldu : [q \rangle \rightarrow \Boolean_\reject^{n+1}$
and $\bldtheta : \Fset_q \rightarrow \Boolean_*^{n+1}$ as follows.
For every $x \in [q \rangle$,
\[
\bldu(x) =
\left\{
\begin{array}{cl}
\bldu'(x) \, {*} & \textrm{if $x \ne y$} \\
\bldu'(y) \, 0   & \textrm{if $x = y$}
\end{array}
\right.
\]
(namely, $\bldu(x)$ is obtained from $\bldu'(x)$
by appending a~$*$ or a~$0$, depending on~$x$),
and for every $f \in \Fset_q$,
\[
\bldtheta(f) =
\left\{
\begin{array}{cl}
\bldtheta'(f) \, {*}          & \textrm{if $f \ne g$} \\
\bldtheta'(\overline{g}) \, 1 & \textrm{if $f = g$}
\end{array}
\right. .
\]
It can be verified that $f(x) = \TCAM(\bldu(x),\bldtheta(f))$
for every $x \in [q \rangle$ and $f \in \Fset_q$, which means
that $\Fset_q$ is $(n+1)$-cell implementable under
Scenario~(\Srd). The result follows from
Proposition~\ref{prop:FsetSdd}.
\end{proof}

The next lemma provides the sufficiency part
in Proposition~\ref{prop:Fset'Sdd}
(and, therefore, also in Proposition~\ref{prop:Fset'Srd} for $q \ge 3$)
for the special case where~$\bldf$ contains exactly one~$0$;
due to symmetry, we can assume that $\bldf = 11 \ldots 110$,
i.e., $f(x) = \LE_q(x,q-2)$.

\begin{lemma}
\label{lem:1110}
Under Scenario~\textup{(\Sdd)},
the subset $\Phi = \Fset_q \setminus \{ \LE_q(\cdot,q-2) \}$
is $n$-cell implementable, whenever
\[
3 \le q \le n+1 .
\]
\end{lemma}

\begin{proof}
We present mappings $\bldu : [q \rangle \rightarrow \Boolean_*^{q-1}$
and $\bldtheta : \Phi \rightarrow \Boolean_*^{q-1}$
for which~(\ref{eq:TCAM}) holds for every $f \in \Phi$.
We take the mapping~$\bldu$ to be
\begin{equation}
\label{eq:1110x=q-1}
\bldu(q-1) = 00 \ldots 0
\end{equation}
and, for $x \in [q-1 \rangle$:
\begin{equation}
\label{eq:1110x!=q-1}
u_j(x) =
\left\{
\begin{array}{ll}

1   & \textrm{if $x = j$} \\
0   & \textrm{if $x = \langle j{+}1 \rangle$} \\
{*} & \textrm{otherwise}
\end{array}
\right. ,
\end{equation}
where $\langle \cdot \rangle = \langle \cdot \rangle_{q-1}$ denotes
the remainder (in $[q-1 \rangle$) modulo $q-1$;
namely, $\bldu(x)$ is obtained from the unary representation of~$x$
by changing all the $0$'s into $*$'s, except the~$0$
that immediately precedes the (unique) $1$ in that representation
(with indexes extended cyclically modulo $q-1$).
Table~\ref{tab:1110u} shows this mapping for $q = n+1 = 5$.
Thus, $\bldu(x)$ contains at least one~$0$
for every $x \in [q \rangle$ (recall that $n \ge 2$).

Turning to the mapping~$\bldtheta$, we select it so that
\begin{equation}
\label{eq:1110f=0}
\bldtheta(\allzero_q) = 11 \ldots 1
\end{equation}
(indeed, $\TCAM(\bldu(x),\bldtheta(\allzero_q)) = 0$
for all $x \in [q \rangle$); for all functions other than
$\allzero_q$ and $\LE_q(\cdot,q-2)$ we let
\begin{equation}
\label{eq:1110f!=0}
\vartheta_j(f) =
\left\{
\begin{array}{ll}
0   & \textrm{if $f(j) = 0$} \\
1   & \textrm{if $f(j) = 1$ and
      $f( \langle j{+}1 \rangle) = f(q{-}1) = 0$} \\
{*} & \textrm{otherwise}
\end{array}
\right.
.
\end{equation}
Specifically,
$\bldtheta(f)$ is obtained from the $(q-1)$-prefix of the truth
table, $\bldf$, of~$f$
by changing into $*$'s all the $1$'s (if $f(q-1) = 1$), or
just the $1$'s that immediately precede other $1$'s (if $f(q-1) = 0$,
with indexes taken modulo $q-1$).
Table~\ref{tab:1110theta} shows this mapping for $q = n+1 = 4$
(note the missing entry from the table, $\bldf = 1110$,
which corresponds to $f(x) = \LE_4(x,2)$).
Thus, $\bldtheta(f)$ will contain at least one~$1$
\ifandonlyif\ $f(q-1) = 0$ and, so,
$\TCAM(\bldu(q-1),\bldtheta(f)) = f(q-1)$.

It remains to show that Eq.~(\ref{eq:TCAM}) holds
when $x \in [q-1 \rangle$ and $f \ne \allzero_q$.
If~$x$ is such that $f(x) = 0$ then,
by~(\ref{eq:1110x!=q-1}) and~(\ref{eq:1110f!=0}),
\[
u_x(x) = 1
\quad \textrm{yet} \quad \vartheta_x(f) = 0
\]
and, so, $\TCAM(\bldu(x),\bldtheta(f)) = 0 = f(x)$.
On the other hand, if $f(x) = 1$ then, for every $j \in [q-1 \rangle$:
\[
u_j(x) = {*}
\quad \textrm{if $j \not\in \{ x, \langle x-1 \rangle \}$}
\]
and
\[
u_j(x) = 1 \;\; \textrm{and} \;\; \vartheta_j(f) \in \{ 1, {*} \}
\quad \textrm{if $j = x$}.
\]
As for the index $j = \langle x-1 \rangle$,
here $u_j(x) = 0$ while $\vartheta_j(f) \ne 1$ in this case
(since $f( \langle j+1 \rangle) = f(x) = 1$).
Hence, $\TCAM(\bldu(x),\bldtheta(f)) = 1 = f(x)$.
\end{proof}

\begin{table}[hbt]
\caption{Mapping $x \mapsto \bldu(x)$ for $q = n+1 = 5$
in the proof of Lemma~\protect\ref{lem:1110}.}
\label{tab:1110u}
\[
\begin{array}{c|c}
x &  \bldu(x) \\
\hline
\raisebox{0ex}[2ex]{}
0      & 1{*}{*}0 \\
1      & 01{*}{*} \\
2      & {*}01{*} \\
3      & {*}{*}01 \\
4      & 0000
\end{array}
\]
\end{table}
\begin{table}[hbt]
\caption{Mapping $f \mapsto \bldtheta(f)$ for $q = n+1 = 4$
in the proof of Lemma~\protect\ref{lem:1110}.}
\label{tab:1110theta}
\[
\begin{array}{c|c}
\bldf  & \bldtheta(f) \\
\hline
\raisebox{0ex}[2ex]{}
0000   & 111           \\
1000   & 100           \\
0100   & 010           \\
1100   & {*}10         \\
0010   & 001           \\
1010   & 10{*}         \\
0110   & 0{*}1         \\
\multicolumn{2}{c}{}
\end{array}
\quad\quad
\begin{array}{c|c}
\bldf  & \bldtheta(f) \\
\hline
\raisebox{0ex}[2ex]{}
0001   & 000           \\
1001   & {*}00         \\
0101   & 0{*}0         \\
1101   & {*}{*}0       \\
0011   & 00{*}         \\
1011   & {*}0{*}       \\
0111   & 0{*}{*}       \\
1111   & {*}{*}{*}
\end{array}
\]
\end{table}

\begin{remark}
\label{rem:1110}
The case $q = n+1 = 2$
was excluded from Lemma~\ref{lem:1110},
since, under Scenario~(\Sdd), the lemma does not hold for the set
\begin{equation}
\label{eq:q=n+1=2}
\Phi = \Fset_2 \setminus \{ \LE_2(\cdot,0) \}
= \{ \allzero_2, \GE_2(\cdot,1), \allone_2 \} .
\end{equation}
Indeed, regardless of how we select $\vartheta(\allzero_2)$,
we would get $\TCAM(u(x),\vartheta(\allzero_2)) = 1$
for at least one $x \in [2 \rangle$.

However, the lemma does hold for the set~(\ref{eq:q=n+1=2})
under Scenario~(\Srd) by taking
\[
u(0) = \reject ,
\quad u(1) = 1
\]
and
\[
\vartheta(\allzero_2) = 0 , \quad
\vartheta(\GE_2(\cdot,1)) = 1 , \quad
\vartheta(\allone_2) = {*} .
\]\qed
\end{remark}

\begin{proof}[Proof of Proposition~\protect\ref{prop:Fset'Sdd}
(sufficiency)]
We show that $q = n+1$ is achievable by induction on~$n$.
For $n = 1$, the only choice for~$g$ is~$\allzero_2$,
in which case we take
\begin{equation}
\label{eq:base-u}
u(0) = 0 ,
\quad
u(1) = 1 ,
\end{equation}
and
\begin{equation}
\label{eq:base-theta}
\vartheta(\LE_2(\cdot,0)) = 0 ,
\quad
\vartheta(\GE_2(\cdot,1)) = 1 ,
\quad
\vartheta(\allone_2) = {*} .
\end{equation}

Turning to the induction step, assume that $q = n+1 \ge 3$
and let $\bldg = (g(x))_{x \in [q \rangle}$ be the truth
table of some function $g \in \Fset_q$
not in $\Eset_q \cup \{ \allone_q \}$.
Without loss of generality we can assume that
\[
\bldg = \underbrace{1 1 \ldots 1}_{t+1} 0 0 \ldots 0 ,
\]
namely, $g(x) = \LE_q(x,t)$, where $t \ne 0, q-1$.
The case $t = q-2$ is covered by
Lemma~\ref{lem:1110}, so we assume hereafter that
$t \le q-3$ (i.e., $g(q-2) = g(q-1) = 0$).

Consider the function $g' \in \Fset_{q-1}$ which
is the restriction of~$g$ to the domain $[q-1 \rangle$.
Then $g' \not\in \Eset_{q-1} \cup \{ \allone_{q-1} \}$
by our assumptions and, so, by the induction hypothesis,
there exist $\bldu' : [q-1 \rangle \rightarrow \Boolean_*^{n-1}$
and $\bldtheta' : \Fset_{q-1} \rightarrow \Boolean_*^{n-1}$ such that
\begin{equation}
\label{eq:TCAMf'}
\TCAM(\bldu'(x),\bldtheta'(f')) = f'(x)
\end{equation}
for every $f' \in \Fset_{q-1} \setminus \{ g' \}$
and $x \in [q-1 \rangle$.

Define the mapping $\bldu : [q \rangle \rightarrow \Boolean_*^n$ by
\begin{equation}
\label{eq:x=q-1}
\bldu(q-1) = {*}{*} \ldots {*}{*}1
\end{equation}
and, for $x \in [q-1 \rangle$:
\begin{equation}
\label{eq:x!=q-1}
\bldu(x) =
\left\{
\begin{array}{lcl}
\bldu'(x) \, 0        & \textrm{if $g(x) = 0$} \\
\bldu'(x) \, {*}      & \textrm{if $g(x) = 1$}
\end{array}
\right. .
\end{equation}

Let $\overline{g} : [q \rangle \rightarrow \Boolean$
be the function which is identical to~$g$ except that
$\overline{g}(q-1) = 1$.
We define the mapping
$\bldtheta : \Fset_q \setminus \{ g \} \rightarrow \Boolean_*^n$ by
\begin{equation}
\label{eq:f=overlineg}
\bldtheta(\overline{g}) = {*}{*} \ldots {*}{*}1
\end{equation}
while for $f \in \Fset_q \setminus \{ g, \overline{g} \}$:
\begin{equation}
\label{eq:f!=overlineg}
\bldtheta(f) =
\left\{
\begin{array}{lcl}
\bldtheta'(f') \, 0   & \textrm{if $f(q-1) = 0$} \\
\bldtheta'(f') \, {*} & \textrm{if $f(q-1) = 1$}
\end{array}
\right. ,
\end{equation}
where $f'$ is the restriction of~$f$ to the domain $[q-1 \rangle$.

We now verify that Eq.~(\ref{eq:TCAM}) holds for
every $x \in [q \rangle$ and
$f \in \Fset_q \setminus \{ g \}$.
For $x = q-1$ and $f = \overline{g}$ we have
\[
\TCAM(\bldu(q-1),\bldtheta(\overline{g}))
\stackrel{\textrm{(\ref{eq:x=q-1})+(\ref{eq:f=overlineg})}}{=}
1 = \overline{g}(q-1) ,
\]
while for $x = q-1$ and $f \ne \overline{g}$,
\[
\TCAM(\bldu(q-1),\bldtheta(f))
\stackrel{\textrm{(\ref{eq:x=q-1})}}{=}
\TCAM(1, \vartheta_{n-1}(f))
\stackrel{\textrm{(\ref{eq:f!=overlineg})}}{=}
f(q-1) .
\]
Assuming now that $x \in [q-1 \rangle$, for $f = \overline{g}$,
\[
\TCAM(\bldu(x),\bldtheta(\overline{g}))
\stackrel{\textrm{(\ref{eq:f=overlineg})}}{=}
\TCAM(u_{n-1}(x), 1)
\stackrel{\textrm{(\ref{eq:x!=q-1})}}{=}
g(x) = \overline{g}(x) ,
\]
while for $f \ne \overline{g}$,
\[
\TCAM(\bldu(x),\bldtheta(f))
\stackrel{\textrm{(\ref{eq:x!=q-1})+(\ref{eq:f!=overlineg})}}{=}
\TCAM(\bldu'(x),\bldtheta'(f'))
\stackrel{\textrm{(\ref{eq:TCAMf'})}}{=}
f'(x) = f(x) .
\]
\end{proof}

\begin{proof}[Proof of Proposition~\protect\ref{prop:Fset'Srd}
(sufficiency)]
The proof is very similar to that of
Proposition~\ref{prop:Fset'Srd}, except that in the induction base
we need to take into account that~$g$ may also be $\LE_2(\cdot,0)$
(or $\GE_2(\cdot,1)$); this case is covered by Remark~\ref{rem:1110}.
Respectively, in the induction step, $g$
can also be $\LE_q(\cdot,0)$.
\end{proof}

\begin{example}
\label{ex:g=0}
When $g = \allzero_q$,
running the recursive definitions
(\ref{eq:x=q-1})--(\ref{eq:f!=overlineg})
with the initial conditions
(\ref{eq:base-u})--(\ref{eq:base-theta})
yields, for every $j \in [q-1 \rangle$,
$x \in [q \rangle$,
and $f \in \Fset_q \setminus \{ \allzero_q \}$:
\[
u_j(x) =
\left\{
\begin{array}{ll}
0   & \textrm{if $x \le j$} \\
1   & \textrm{if $x = j+1$} \\
{*} & \textrm{if $x > j+1$}
\end{array}
\right.
\]
and
\[
\vartheta_j(f) =
\left\{
\arraycolsep0.9ex
\begin{array}{ll}
0   & \textrm{if $f(j{+}1) = 0$ and $f(x) = 1$ for some $x \le j$} \\
1   & \textrm{if $f(j{+}1) = 1$ and $f(x) = 0$ for all $x \le j$} \\
{*} & \textrm{otherwise}
\end{array}
\right.
\!\!\!\! .
\]
In words, $\bldtheta(f)$ is obtained from the truth table of~$f$
by changing all the $0$'s that precede the first~$1$
and all the $1$'s that succeed it into $*$'s, and then deleting
the first entry (which corresponds to $x = 0$).
Table~\ref{tab:Fset!=0Sdd}
shows the mappings $x \mapsto \bldu(x)$
and $f \mapsto \bldtheta(f)$ for $q = 4$.\qed
\begin{table}[hbt]
\caption{Mappings
for $\Fset_4 \setminus \{ \allzero_4 \}$
that attain $q = n+1 = 4$ under Scenario~(\Sdd).}
\label{tab:Fset!=0Sdd}
\[
\begin{array}{c|c}
x      & \bldu(x)     \\
\hline
\raisebox{0ex}[2ex]{}
0      & 000          \\
1      & 100          \\
2      & {*}10        \\
3      & {*}{*}1      \\
\multicolumn{1}{c}{} &
\end{array}
\quad\quad\quad
\begin{array}{c|c}
\bldf  & \bldtheta(f) \\
\hline
\raisebox{0ex}[2ex]{}
1000   & 000           \\
0100   & 100           \\
1100   & {*}00         \\
0010   & {*}10         \\
1010   & 0{*}0         \\
\end{array}
\quad
\begin{array}{c|c}
\bldf  & \bldtheta(f) \\
\hline
\raisebox{0ex}[2ex]{}
0110   & 1{*}0         \\
1110   & {*}{*}0       \\
0001   & {*}{*}1       \\
1001   & 00{*}         \\
0101   & 10{*}         \\
\end{array}
\quad
\begin{array}{c|c}
\bldf  & \bldtheta(f) \\
\hline
\raisebox{0ex}[2ex]{}
1101   & {*}0{*}       \\
0011   & {*}1{*}       \\
1011   & 0{*}{*}       \\
0111   & 1{*}{*}       \\
1111   & {*}{*}{*}
\end{array}
\]
\end{table}
\end{example}

When $g \ne \allzero_q$ we can assume that
$g(x) = \LE_q(x,t)$ for some $t \in [1:q-1 \rangle$
(or $t \in [q-1 \rangle$ for Scenario~(\Srd)).
We run the recursive definitions
(\ref{eq:x=q-1})--(\ref{eq:f!=overlineg})
yet now with the initial conditions
(\ref{eq:1110x=q-1})--(\ref{eq:1110f!=0})
when stated for $q = n+1 = t+2$.
Table~\ref{tab:Fset!=t=3}
shows the resulting mapping $x \mapsto \bldu(x)$
for $q = 8$ and $t = 3$.
The entries $u_j(x)$ that correspond to
$(x,j) \in [t+2 \rangle \times [t+1 \rangle$
are determined by the initial conditions
and are the same as in Table~\ref{tab:1110u},
while the entries that correspond to $(x+t+1,j+t+1)$
for $(x,j) \in [q-t-1 \rangle \times [q-t-2 \rangle$
are the same as in Table~\ref{tab:Fset!=0Sdd}.

\begin{table}[hbt]
\caption{Mapping $x \mapsto \bldu(x)$
that attains $q = n+1$ for $q = 8$ and $t = 3$.\vspace{-2ex}}
\label{tab:Fset!=t=3}
\[
\arraycolsep0.5ex
\begin{array}{c|r|l}
\quad x \quad     & \multicolumn{2}{c}{\bldu(x)}     \\
\hline
\raisebox{0ex}[2ex]{}
0      & 1{*}{*}0     & {*}{*}{*} \\
1      & 01{*}{*}     & {*}{*}{*} \\
2      & {*}01{*}     & {*}{*}{*} \\
3      & {*}{*}01     & {*}{*}{*} \\
\cline{3-3}
4      & 0000         & 000       \\
\cline{2-2}
5      & {*}{*}{*}{*} & 100       \\
6      & {*}{*}{*}{*} & {*}10     \\
7      & {*}{*}{*}{*} & {*}{*}1
\end{array}
\]
\end{table}

\begin{proof}[Proof of Proposition~\protect\ref{prop:Fsetstar}]
Suppose that $\Fset_q^\star$
is $n$-cell implementable with mappings
$\bldu : [q \rangle \rightarrow \Boolean_*^n$
and $\bldtheta : \Fset_q^\star \rightarrow \Boolean_*^n$.
Since $q \ge 3$, for any two distinct elements $x, x' \in [q \rangle$
there exists a function $f \in \Fset_q^\star$ such that
$f(x) = 1$ yet $f(x') = 0$.
Hence, the mapping~$\bldu$
has to be injective and its set of images must form
an antichain in $\Boolean_*^n$ (see Section~\ref{sec:Eset}).

For any element $t \in [q \rangle$
let $t'$ be a particular element in $[q \rangle \setminus \{ t \}$
(say, $t' = 1$ when $t = 0$, and $t' = 0$ otherwise).
By the antichain property it follows that
there is at least one index $\ell = \ell(t) \in [n \rangle$
such that $u_\ell(t) \npreceq u_\ell(t')$
and, in particular,
$u_\ell(t) \ne u_\ell(t')$,
$u_\ell(t) \ne \reject$, and $u_\ell(t') \ne {*}$.
We define $b(t) \in \Boolean$ by
\[
b(t)
=
\left\{
\begin{array}{ccl}
u_\ell(t)    && \textrm{if $u_\ell(t) \in \Boolean$} \\
1 &&
\textrm{if $u_\ell(t) = {*}$ and $u_\ell(t') \in \{ 0, \reject \}$} \\
0 &&
\textrm{if $u_\ell(t) = {*}$ and $u_\ell(t') = 1$}
\end{array}
\right.
.
\]
Notice that $\TCAM(u_\ell(t),b(t)) = 1$ yet
$\TCAM(u_\ell(t'),b(t)) = 0$.
We also define $f_t : [q \rangle \rightarrow \Boolean$ to be
the function in $\Fset_q^\star$
which evaluates to~$1$ when $x \in \{ t, t' \}$ and to~$0$ otherwise.

We now extend the mapping~$\bldtheta$ to the domain $\Fset_q$ by
\[
\bldtheta(\allone_q) = {*}{*} \ldots {*}
\]
and, for every $t \in [q \rangle$ and $j \in [n \rangle$,
\[
\vartheta_j(\E_q(\cdot,t))
=
\left\{
\begin{array}{ccl}
\vartheta_j(f_t)  && \textrm{if $j \ne \ell(t)$} \\
b(t)              && \textrm{if $j = \ell(t)$}
\end{array}
\right.
.
\]
It can be readily verified that Eq.~(\ref{eq:TCAM})
holds for the added functions $\allone_q$
and $\E_q(\cdot,t)$, $t \in [q \rangle$,
and, therefore, for the whole set $\Fset_q$.
Hence, $\Fset_q$ is $n$-cell implementable.
\end{proof}
\fi

\ifANONYMOUS
\else
\section*{Acknowledgment}

\ifPAGELIMIT
    This work resulted from many stimulating discussions that
    I had with Giacomo Pedretti from Hewlett Packard Labs.
\else
I would like to express my thanks
to Giacomo Pedretti from Hewlett Packard Labs.
This work resulted from many stimulating discussions that
I had with him during my visit at Labs.
\fi
Thanks are also due to Luca Buoanno and Cat Graves.
\fi


\begin{thebibliography}{99}
\bibitem{Anderson}
    \bibauthor{I. Anderson,}
    \bibbook{Combinatorics of Finite Sets,}
    Dover Publications, Mineola, New York, 1987.
\bibitem{AB}
    \bibauthor{M. Anthony, N. Biggs,}
    \bibbook{Computational Learning Theory,}
    Cambridge University Press, Cambridge, UK, 1992.
\ifPAGELIMIT
\else
\bibitem{BKKS}
    \bibauthor{A. Birkendorf, N. Klasner, C. Kuhlmann, H.U. Simon,}
    \bibpaper{Structural results about exact learning with
    unspecified attribute values,}
    \bibperiodical{J. Comput.\ Syst.\ Sci.,}
    60 (2000), 258--277.
\fi
\bibitem{BHH}
    \bibauthor{A. Bremler-Barr, D. Hay, D. Hendler,}
    \bibpaper{Layered interval codes for TCAM-based classification,}
    \bibperiodical{Proc.\ 28th IEEE Conf.\ on Computer Communication
    (INFOCOM 2009),}
    Rio de Janeiro, Brazil, 2009, pp.~1305--1313.
\bibitem{FRY}
    \bibauthor{W. Fraczak, W. Rytter, M. Yazdani,}
    \bibpaper{Matching integer intervals by minimal sets of
    binary words with don't cares,}
    \bibperiodical{Proc.\ 19th Annual Symp.\ on
    Combinatorial Pattern Matching (CPM 2008),}
    Pisa, Italy, 2008, pp.~217--228.
\bibitem{GG}
    \bibauthor{R. Govindaraj, S. Ghosh,}
    \bibpaper{Design and analysis of STTRAM-based
    ternary content addressable memory cell,}
    \bibperiodical{ACM J. Emerg.\ Technol.\ Comput.\ Syst.,}
    13 (2017), Article 52, 1--22.
\bibitem{KTFY}
    \bibauthor{G. Kasai, Y. Takarabe, K. Furumi, M. Yoneda,}
    \bibpaper{200MH/200MSPS 3.2W at 1.5V Vdd, 9.4Mbits ternary CAM
    with new charge injection match detect circuits and bank selection
    scheme,}
    \bibperiodical{Proc.\ IEEE 2003 Custom Integrated Circuits Conf.,}
    San Jose, CA, 2003, pp.~387--380.
\bibitem{LGSMFPS}
    \bibauthor{C. Li, C.E. Graves, X. Sheng, D. Miller, M. Foltin,
    G. Pedretti, J.P. Strachan,}
    \bibpaper{Analog content-addressable memories with memristors,}
    \bibperiodical{Nat.\ Commun.,} 11 (2020), \#1638.
\bibitem{Lubell}
    \bibauthor{D. Lubell,}
    \bibpaper{A short proof of Sperner's theorem,}
    \bibperiodical{J. Comb.\ Theory,} 1 (1966), 299.
\bibitem{MS}
    \bibauthor{F.J. MacWilliams, N.J.A. Sloane,}
    \bibbook{The Theory of Error-Correcting Codes,}
    North-Holland, Amsterdam, 1977.
\bibitem{NS}
    \bibauthor{T. Natschl\"{a}ger, M. Schmitt,}
    \bibpaper{Exact VC-dimension of Boolean monomials,}
    \bibperiodical{Inf.\ Process.\ Lett.,} 59 (1996), 19--20.
\bibitem{PGLSSFMS}
    \bibauthor{G. Pedretti, C.E. Graves, C. Li, S. Serebryakov,
    X. Sheng, M. Foltin, R. Mao, J.P. Strachan,}
    \bibpaper{Tree-based machine learning performed in-memory
    with memristive analog CAM,}
    \bibperiodical{Nat.\ Commun.,} 12 (2021), \#5806.
\bibitem{YSMB}
    \bibauthor{X. Yang, S. Sezer, J. McCanny, D. Burns,}
    \bibpaper{A versatile content addressable memory architecture,}
    \bibperiodical{Proc.\ 20th Anniversary IEEE Int'l SOC Conf.,}
    Hsinchu, Taiwan, 2007, pp.~215--218.
\end{thebibliography}
\end{document}